\documentclass[a4paper,12pt]{article}
\usepackage[utf8]{inputenc}
\usepackage{amssymb}
\usepackage{amsthm}
\usepackage{amsfonts}
\usepackage{color}
\usepackage{amsmath}
\usepackage{algorithm}
\usepackage{algorithmic}
\usepackage{bm}

\usepackage[font=small]{caption}
\usepackage{graphicx}
\usepackage[toc,page]{appendix}
\usepackage{dsfont}
\usepackage{setspace}
\usepackage{verbatim}
\usepackage{enumerate}
\usepackage[stable]{footmisc}
\usepackage{pdflscape}
\usepackage{lipsum}
\usepackage{xspace}
\usepackage[margin=1.25in]{geometry}
\usepackage[gen]{eurosym}
\usepackage{epstopdf}
\usepackage{subcaption, mwe}
\usepackage[numbers]{natbib}
\newtheorem{remark}{Remark}
\newtheorem{lem}{Lemma}
\newtheorem{definition}{Definition}
\newtheorem{proposition}{Proposition}
\newtheorem{theorem}{Theorem}
\newtheorem*{theorem*}{Theorem}
\newtheorem*{note*}{Note}
\begin{document}

\newpage	

\pagenumbering{arabic}
\title{\textbf{\Large Bounds on Multi-asset Derivatives\\
 via Neural Networks}}
		
\author{Luca De Gennaro Aquino\footnote{Corresponding author. Department of Accounting, Law and Finance, Grenoble Ecole de Management, F-38000 Grenoble, France. (Email: \texttt{luca.degennaroaquino@grenoble-em.com}). } \\        \and
       Carole Bernard\footnote{Department of Accounting, Law and Finance, Grenoble Ecole de Management, F-38000 Grenoble, and Department of Economics and Political Sciences at Vrije Universiteit Brussel (VUB).
             (Email: \texttt{carole.bernard@grenoble-em.com}).}
\footnote{We thank the participants of the Vienna Congress on Mathematical Finance in September 2019 and of the RiO 2019 annual meeting, as well as Lamya Kermiche, Fran\c{c}ois Desmoulins Lebeault, Stephan Eckstein and Michael Kupper for helpful comments on the first draft of the paper. }}
\date{\today}
\maketitle	
\vspace{2cm}

\begin{abstract}
 Using neural networks, we compute bounds on the prices of multi-asset derivatives given information on prices of related payoffs. As a main example, we focus on European basket options and include information on the prices of other similar options, such as spread options and/or basket options on subindices. We show that, in most cases, adding further constraints gives rise to bounds that are considerably tighter. Our approach follows the literature on constrained optimal transport and, in particular, builds on Eckstein and Kupper \cite{eckstein2018computation}.
\end{abstract}


\newpage
\section{Introduction}
In this paper, we consider the problem of computing model-independent upper and lower bounds on the price of multi-asset derivatives, given prices on related payoffs. We  devote our attention mainly to European basket options, although our methodology can be applied also to other types of multivariate options.

Basket options are options written on a linear combination of some underlying assets. For instance, the payoff for a call is given by the positive difference between the weighted sum of the prices of the different assets in the basket and the exercise price. 
Characterizing closed form formulas for the price of basket options (and in general multi-asset derivatives) is a complicated task. The difficulty stems from the unavailability of the distribution for the weighted sum of the (possibly correlated) underlying assets. For instance, in the higher-dimensional Black-Scholes setting, this would require the distribution of the sum of correlated log-normals, which is not available in explicit form. Within this framework, approximate or partially explicit formulas were obtained by Carmona and Durrleman \cite{carmona2005generalizing} and Deelstra et al. \cite{deelstra2004pricing}. Despite this complication, different approaches for pricing basket options have been provided in the literature, so we refer the interested readers to the papers by Milevsky and Posner \cite{milevsky1998closed}, Ju \cite{ju2002pricing}, Brigo et al. \cite{brigo2004approximated}, Borovkova et al. \cite{borovkova2007closed}, Linders and Stassen \cite{linders2016multivariate}. A review of some of these approaches can be found in the chapter by Krekel et al. \cite{krekel2006analysis}.

The wide range of quantitative models that have arised in the financial literature and are used for pricing and hedging purposes (stochastic volatility models, local volatility models, Lévy-type models and so on) has given rise to the so-called ``model uncertainty", that is, the uncertainty on the choice of the model itself and its impact on the pricing of derivative instruments. A thorough scrutiny of this topic can be found in Cont \cite{cont2006model}. In his paper, the author overviews some of the tools that have been used in risk management for mitigating model risk (mainly, Bayesian model averaging (Hoeting et al. \cite{hoeting1999bayesian}) and \textit{maxmin}-type approaches (Epstein and Wang \cite{epstein1995uncertainty})), besides proposing two new frameworks based on coherent and convex risk measures.

With this scenario in mind, we are concerned with an integrative strategy for assessing plausibility of prices and hedges of derivative payoffs obtained via some parametric model, that is, computing non-parametric (i.e. model-independent) upper and lower bounds for these quantities. This complementary approach does not rely on any assumption on the dynamic of the underlying assets, but rather considers the class of all models that are consistent with observed market data.  
Previous works that addressed this problem for path-independent multi-asset options include the papers by Hobson et al. \cite{hobson2005static2, hobson2005static}, Laurence and Wang \cite{laurence2004s, laurence2005sharp}, d'Aspremont and El Ghaoui \cite{d2006static}, Chen et al. \cite{chen2008static} and Pe{\~{n}}a et al. \cite{pena2010static}. In all  of these papers, the computation of static-arbitrage bounds typically required the solution of some infinite-dimensional linear programming problem. We also mention Hobson \cite{hobson2011skorokhod}, which pointed out and discussed the relation of this problem to the Skorokhod embedding problem, and Caldana et al. \cite{caldana2016general}, who found basket options pricing bounds for a general class of continuous-time models assuming knowledge of the joint characteristic function of the log-returns of the assets.  For path-dependent options, Beiglb\"{o}ck et al. \cite{beiglbock2013model} and Galichon et al. \cite{galichon2014stochastic} obtained robust bounds by establishing a dual version of the problem using arguments from the theory of martingale optimal transport. This dual formulation ended up being intimately linked to the construction of a semi-static hedging portfolio.

For other applications in risk management under uncertainty, we refer to the book by McNeil et al. \cite{mcneil2005quantitative} and, among all, to the papers by Embrechts et al. \cite{embrechts2013model}, Bernard et al. \cite{bernard2014risk}, Bernard et al. \cite{bernard2020model} and references therein.

In this paper, while keeping an optimal transport setting, we aim to compute numerically bounds on multi-asset options by means of neural networks, an approach which has been initially developed by \cite{eckstein2018computation}. In the spirit of d'Aspremont and El Ghaoui \cite{d2006static}, we aim at tightening existing bounds by including knowledge on the prices of other multivariate derivatives written on the same assets, beyond the information on the marginals. Our setting is also closely related to the papers by Eckstein et al. \cite{eckstein2019robust}, who obtained bounds on multi-asset options by including prices on single-asset options at different maturities, and L\"{u}tkebohmert and Sester \cite{lutkebohmert2019tightening}, who included information on the covariance or correlation between assets. Also, more recently, Neufeld et al. \cite{neufeld2020model} reformulated the problem as a (non-transport) linear semi-infinite programming problem and provided cutting-plane algorithms for model-free bounds using option-implied information.

The main contribution of this paper is to show that neural networks allow to solve for model uncertainty on prices of multi-asset derivatives in very general and possibly higher-dimensional settings, and to illustrate the methodology with specific basket and spread options.

We also extend the work of Tavin \cite{tavin2015detection}, whose idea is to parametrise the set of copulas using a dense parametric family of copulas (the Bernstein copulas) and derive bounds on prices of multi-asset derivatives using this family of copulas constrained to match a certain number of constraints (existing prices). In our paper, the algorithm seeks a copula in a non-parametric way and without specifying from which family it must be derived. It is thus a more general approach than Tavin's, but allows to recover his results.

As a further note, an important application of our study is the detection of arbitrage opportunities in the market. In practice, a certain number of multi-asset derivatives are priced. A prime concern for a trader who wants to introduce a new multi-asset derivative in the market is to not introduce arbitrage. To do so, one needs to determine bounds on the price of the multi-asset derivative that take into account all the existing information. Should a derivative in the market have a price outside its arbitrage bounds, this would mean that such derivative is mispriced. Detecting the arbitrage is relatively straightforward using our methodology. However, taking advantage of this arbitrage may not always be straightforward as multi-asset derivatives are typically hard to replicate. 

This paper builds on two strands of literature, i.e., optimal transport (denoted hereafter by OT) and artificial neural networks. OT dates back to the seminal work in 1781 by French geometer Monge \cite{monge1781memoire}. Although it was initially formalised as a problem of cost-efficient transportation of mineral resources, subsequently it has been extensively applied in numerous fields and under different settings, from economics to quantum physics. In the classic two-dimensional case, considering some cost function $c(x,y)$ on $\mathcal{X} \times \mathcal{Y}$, $\mathcal{X}$ and $\mathcal{Y}$ representing two suitable spaces, the problem can be simply stated in terms of the following minimization:
\begin{equation} \label{OTgen}
\inf_{\pi \in \Pi(\mu_{1}, \mu_{2})}\int_{\mathcal{X} \times \mathcal{Y}}c(x,y)\, d\pi(x,y),
\end{equation}
where $\Pi(\mu_{1}, \mu_{2})$ is the set of all joint probability measures with given marginals $\mu_{1}$ and $\mu_2$ on $\mathcal{X}$ and $\mathcal{Y}$, respectively. Such joint measures are called transportation plans and, trivially, those achieving the infimum are called optimal transportation plans. We will see in Section \ref{setting} how the problem of computing bounds on derivative payoffs can be framed in similar terms. For a detailed treatment of optimal transport theory, well-known reference books include Rachev and R\"{u}schendorf \cite{rachev1998mass, rachev2006mass} and Villani \cite{villani2003topics, villani2008optimal}. For an overview of computational methods for OT problems, see instead Peyré and Cuturi \cite{peyre2019computational}. In Section \ref{MK with constraints}, we  review some results by Zaev \cite{zaev2015monge} on a (linearly) constrained version of the classic Monge-Kantorovich duality, which are a key ingredient to justify the approach in Sections \ref{NeuralNetworks}-\ref{applications}. Interestingly, this framework relates to the martingale version of the Monge-Kantorovich problem, for which we refer to the papers by Hobson et al. \cite{hobson2011skorokhod} and Beiglb\"{o}ck et al. \cite{beiglbock2013model, beiglbock2016problem, beiglbock2017complete}. 

On the other hand, artificial neural networks represent one of the most popular algorithms in machine learning. Due to their flexibility and capability of performing nonlinear modeling without a priori assumptions or model specifications, they have been widely applied in the last years on a plethora of tasks, from computer vision to speech recognition as well as in  econometrics and mathematical finance. In fact, neural networks have been shown to exhibit universal approximation properties (Hornik et al. \cite{hornik1989multilayer}, Hornik \cite{hornik1991approximation, hornik1993some}), which make them suitable for approximating any continuous functions with arbitrary accuracy. We  discuss further on this point in Section \ref{NeuralNetworks}. For a review of applications of neural networks specifically on option pricing and hedging, see Ruf and Wang \cite{ruf2019neural}. Among others, promising applications have been proposed by Buehler et al. \cite{buehler2019deep}, in which the authors introduced a reinforcement learning framework for hedging derivatives in markets with frictions, by  Lu \cite{lu2017agent} and Du et al. \cite{du2016algorithm}, who implemented algorithmic trading strategies using recurrent neural networks, and by Sirignano and Cont \cite{sirignano2019universal}, who performed a large-scale analysis of the mechanism of price formation using high-frequency data.

The paper is organised as follows. In Section \ref{setting}, we  provide the setting and define the problem of computing bounds on basket options in terms of optimal transport (OT). In Section \ref{MK with constraints},  we  review some results by Zaev \cite{zaev2015monge} on OT theory with additional linear constraints. In Section \ref{NeuralNetworks}, we  introduce the approach in Eckstein and Kupper \cite{eckstein2018computation}  and give some background on neural networks. In Section \ref{applications}, we  present some applications on basket options. Some proofs and complementary numerical results are postponed to the appendices.

\section{Setting} \label{setting}
 Consider a one-period financial market consisting of $d \geq 2$ risky assets, modeled as a $\mathbb{R}^{d}_{+}$-valued random vector $( S_{1}(t),\dots,S_{d}(t) )_{t \in \{0,T\}}$ on a probability space $\left( \Omega, \mathcal{F}, \mathbb{P} \right)$. Since we essentially deal with path-independent derivatives, we henceforth omit the reference to time in the notation.\footnote{In Section \ref{PaDe}, we illustrate how the methodology also applies on path-dependent derivatives and reintroduce the time notation at this point. } We assume that the financial market does not allow for arbitrage opportunities, so that the existence of at least one risk-neutral probability measure $\mathbb{Q}$ equivalent to the real-world probability measure $\mathbb{P}$ is guaranteed.  
Let us denote by $C(K, (\alpha_{i})_{i \in J}  )$ the payoff of a basket call option with strike $K \in\mathbb{R} $ and weights $(\alpha_{i})_{ i \in J}$, $\alpha_{i} \neq 0,$ on $ J \subseteq I = \{1,2,\dots,d \}$ assets in our market. Such payoff is given by:
\begin{equation*}
C(K, (\alpha_{i})_{i \in J}  ) = \left( \sum_{i \in J}\alpha_{i}S_{i} - K \right)^{+}.
\end{equation*}    

Notice that when the weights are allowed to be negative these payoffs also include spread options. Also, when the basket only concerns one asset, the basket option reduces to a standard European option written on a single underlying.
Furthermore, we assume that options on a single asset are priced for each strike $K$, so that the distribution of the underlying is fully characterised (Breeden and Litzenberger \cite{breeden1978prices}). In other words, for all $i \in I$, we will assume $X_i \sim \mu_{i}$, for some distributions $\mu_{i}$. 

Let us now define the following system of subindices:
$(\bm{J}^{1},\bm{\alpha}^{1}),\dots,(\bm{J}^{n},\bm{\alpha}^{n})$,  where $\bm{J}^{k} := \{j_{1},\dots,j_{m_{k}}\} \subset \{1,\dots,d\}$ and $\bm{\alpha}^{k} := \left( \alpha_{1}^{k},\dots,\alpha_{m_{k}}^{k} \right)$. Here $n$ represents the number of subindices that we want to consider and $ 1 \leq m_{k} \leq d, k = 1,\dots,n$, the number of assets included in the $k$-th subindex. Also, denote as $\bm{S}^{(\bm{J}^{k},\bm{\alpha}^{k})}$ the linear combination of the $m_{k}$ assets in the $k$-th subindex with weights $\bm{\alpha}^{k}$:
\begin{equation*}
\bm{S}^{(\bm{J}^{k},\bm{\alpha}^{k})} := \sum_{j \in \bm{J}^{k}}\alpha_{j}^{k}S_{j}
, \quad k = 1,\dots,n.
\end{equation*}
\subsection{Formulation of the problem}
Assuming that we know the prices $(p_{k})_{k=1,\dots,n}$ of the call options on each subindex, we want to solve the problem of computing upper and lower bounds on the price of a European basket option written on $(S_{i})_{ i \in J}$, with maturity $T$, strike $K$ and weights $(\alpha_{i})_{ i \in J}$. For the upper bound on the basket call, this problem can be written as:
\begin{equation} \label{supBasket}
\begin{split}
\sup_{\pi \in \mathcal{Q}}\;  & \mathbb{E}_{\pi}\left[ \left( \sum_{i \in J}\alpha_{i}S_{i} - K \right)^{+} \right],\\
\mbox{subject to} \;  & S_{i} \sim \mu_{i}, \quad i \in J \\
& \mathbb{E}_{\pi}\left[\left(\bm{S}^{(\bm{J}^{k},\bm{\alpha}^{k})} -K_{k} \right)^{+}\right] = p_{k}, \quad k = 1,\dots,n,
\end{split}
\end{equation}
where the supremum runs over all probability measures $\pi$ with margins $\mu_{i}, i \in J$, in a class of probability measures $\mathcal{Q}$ that are consistent with the given set of observed prices $p_{k}$. Equivalently, the problem of finding the lower bound on a basket call option can be written as:
\begin{equation} \label{infBasket}
\begin{split}
\inf_{\pi \in \mathcal{Q}} \; & \mathbb{E}_{\pi}\left[ \left( \sum_{i \in J}\alpha_{i}S_{i} - K \right)^{+} \right],\\
\mbox{subject to:} \; \; & S_{i} \sim \mu_{i}, \quad i \in J \\
&  \mathbb{E}_{\pi}\left[\left(\bm{S}^{(\bm{J}^{k},\bm{\alpha}^{k})} -K_{k} \right)^{+}\right] = p_{k}, \quad k = 1,\dots,n.
\end{split}
\end{equation}
Notice that, for now, we consider that all the options have the same maturity. Also, for the ease of presentation, we include information only on call options on the subindices, but we could have also included information on put options or other payoffs, which only depend on the distribution of the assets at time $T$. Furthermore, the problem is stated with the price of a basket call. Other options payoffs, including put options, can be studied similarly. As it is stated, this problem is similar to the problem in Eq. (1) in d'Aspremont and El Ghaoui \cite{d2006static}, which the authors solved via linear programming or, for some special cases, in closed-forms. 

In this paper, we  take a  different approach, which follows from the  simple observation that problems  (\ref{supBasket}) and (\ref{infBasket}) can be represented as a higher-dimensional version of the optimal transport problem in (\ref{OTgen}), given additional information (in the form of constraints). 
In very general terms, defining as $f$ the payoff of a path-independent derivative and as $(w_{k})_{k = 1,\dots,n}$ the payoffs of the options whose prices are assumed to be known (and equal to $(p_{k})_{k = 1,\dots,n}$), the supremum in (\ref{supBasket}) can be written as follows:
\begin{equation} \label{constrainedOTgen}
\begin{split}
 \sup_{\pi \in \mathcal{Q}} \;& \int f \,d\pi, \\
\mbox{subject to:}\; \;  & ({\Pr}_{i})_{\#}\pi = \mu_{i}, \quad i \in J, \\
& \int w_{k}d\pi = p_{k}, \quad  k=1,\dots,n,
\end{split}
\end{equation}
where $ ({\Pr}_{i})_{\#}\pi$ can be interpreted as the $i$-th marginal of $\pi$. We  define all these ingredients more rigorously in the following section. 
With obvious modifications, a similar equivalent representation holds for the infimum in (\ref{infBasket}).

In the next section, we temporarily abstract ourselves from financial applications and review some results on the problem in (\ref{constrainedOTgen}). In particular, we are interested in obtaining its dual formulation and proving absence of duality gap.
\section{Optimal transport with linear constraints} \label{MK with constraints}
In this section, we recall the optimal transport problem with linear constraints in Zaev \cite{zaev2015monge} in order to ultimately solve the constrained optimization problem \eqref{constrainedOTgen}. All proofs and background results are postponed to Appendix \ref{appendix_Kantorovic}.

Let $\mathcal{X}_{1}, \dots,  \mathcal{X}_{n}$ be Polish spaces with respective Borel $\sigma$-algebras $\mathcal{B}(\mathcal{X}_{i})$, for $i = 1, \dots, n$, and  $\mathcal{X} = \mathcal{X}_{1} \times \cdots \times \mathcal{X}_{n}$. Consider $\mu_{1},\dots,\mu_{n}$ fixed probability measures on $\mathcal{X}_{1}, \dots,  \mathcal{X}_{n}$, with $\mu = \left(\mu_{1},\dots,\mu_{n} \right) $ being a $n$-tuple of such measures. We denote by $\mathcal{P}(\mathcal{X})$ the set of Borel probability measures over $\mathcal{X}$ and by $\Pi(\mu)$ the set of measures on $\mathcal{X}$ with given marginals. Both sets are equipped with the topology of weak convergence.
Let us introduce the functional spaces
\begin{equation*}
	C_{L}(\mu_{i}) := \left\{ h \in L^{1}(\mathcal{X}_{i},\mu_{i}) \cap C(\mathcal{X}_{i})  \right\}
\end{equation*} 
of continuous absolutely integrable functions on $\mathcal{X}_{i}$ (denoted by $C(\mathcal{X}_{i})$) with topology induced by $L^{1}(\mathcal{X}_{i},\mu_{i})$ norm, and $C_{L}(\mu)$ as a subset of the continuous absolutely integrable functions on $\mathcal{X}$, 
\begin{equation*}
\begin{split}
& C_{L}(\mu) := \left\{ c \in C(\mathcal{X}): \exists \, h \in \mathcal{H} \quad \mbox{s.t.} \; \vert c \vert \leq h \right\}, \\
& \mathcal{H} = \left\{ h \in C(\mathcal{X}): h(x_1, \dots, x_n) = \sum_{i=1}^{n}h_i(x_i) \, \mbox{for all } (x_1, \dots,x_n) \in \mathcal{X} \, \mbox{and } h_i \in C_L(\mu_i)  \right\}.
\end{split}	
\end{equation*}
We equip $C_{L}(\mu)$ with the following seminorm:
\begin{equation} \label{seminorm}
\Vert c \Vert_{L} := \sup_{\pi \in \Pi(\mu)}\int_{\mathcal{X}} \vert c \vert \, d\pi.
\end{equation}
The proof that $\Vert \cdot \Vert_{L}$ is a well-defined seminorm is reported in Appendix \ref{well-definiteness}.

Also, let us fix an arbitrary subspace $W \subset C_{L}(\mu)$. 
Our objective is to solve the following (constrained) Monge-Kantorovich optimization problem:
\begin{equation}
\inf_{\pi \in \Pi(\mu)} \left\{\int_{\mathcal{X}} f \, d\pi : \int_{\mathcal{X}} w \, d\pi = 0 \quad \forall w \in W    \right\},
\end{equation}
for some function  $C_{L}(\mu) \ni f: \mathcal{X} \to \mathbb{R}$. By denoting as 
\begin{equation} \label{Pi_w}
\Pi_{W}(\mu) := \left\{ \pi \in \mathcal{P}(\mathcal{X}): \int_{\mathcal{X}} w \, d\pi = 0 \quad \forall w \in W, \; {\Pr}_{\#}\pi = \mu \right\}
\end{equation}
the set of optimal transport plans $\pi$ such that $\pi\vert_{W}=0$, with $\Pr_{\#}\pi = \mu$ denoting the natural projection of $\pi$ from $\mathcal{X}$ on the tuple of spaces $(\mathcal{X}_{1}, \dots,  \mathcal{X}_{n})$,  we can formulate the problem in the following more compact way:
\begin{equation}
\inf_{\pi \in \Pi_{W}(\mu)} \int_{\mathcal{X}} f \, d\pi  .
\end{equation}
Notice that, unlike the setting of optimal transportation without additional contraints, the set $\Pi_{W}(\mu)$ is not guaranteed to be non-empty. If we assume that such set is indeed non-empty, by Theorem 4.1 in Villani \cite{villani2008optimal}, the existence of an optimal transport plan follows from: (a) compactness of $\Pi_{W}(\mu)$, (b) lower semi-continuity of the functional $\pi \to \int_{\mathcal{X}}c \, d\pi$. 
These two properties are assessed in Appendix \ref{compactness-continuity}.

We have now all the elements to state the following generalization of the well-known Monge-Kantorovich duality for the case with additional linear constraints. The full proof is contained in Appendices \ref{lemmas-appendix}-\ref{proof-Kantorovich-constraints}.
\begin{theorem} \label{constrained Kantorovich duality}
Let $\mathcal{X}_{1}, \dots,  \mathcal{X}_{n}$ be Polish spaces, $\mathcal{X} = \mathcal{X}_{1} \times \cdots \times  \mathcal{X}_{n}$, $\mu = (\mu_{i} \in \mathcal{P}(\mathcal{X}_{i}))_{i = 1, \dots, n}$. Also, let $W$ be a subspace of $C_{L}(\mu)$ and $f \in C_{L}(\mu)$. Then,
\begin{equation*}
\inf_{\pi \in \Pi_{W}(\mu)} \int_{\mathcal{X}} f \, d\pi = \sup_{\substack{h + w \leq f \\ h \in \mathcal{H}, \, w \in W}}\sum_{i=1}^{n}\int_{\mathcal{X}_{i}}h_{i}\, d\mu_{i},
\end{equation*}
where $ h(x_1, \dots, x_n) := \sum_{i=1}^{n}h_{i}(x_i)$.
\end{theorem}

\section{Problem penalization and neural networks} \label{NeuralNetworks}
In this section, we modify the approach in Eckstein and Kupper \cite{eckstein2018computation} for solving optimal transportation problems via penalization. All proofs are obtained by adapting the proofs in Sections 2-5 and Appendix A of Eckstein and Kupper \cite{eckstein2018computation}. 

\subsection{Duality results}

As before, let $\mathcal{X}_{1}, \dots,  \mathcal{X}_{n}$ be Polish spaces, $\mathcal{X} = \mathcal{X}_{1} \times \cdots \times  \mathcal{X}_{n}$, $\mu = (\mu_{i} \in \mathcal{P}(\mathcal{X}_{i}))_{i = 1, \dots, n}$, $f \in C_{L}(\mu)$. Here we  consider the maximization problem
\begin{equation} \label{sup pi}
\begin{split}
\phi(f) := \sup_{\pi \in \Pi(\mu)} & \int_{\mathcal{X}}f\, d\pi, \\
\mbox{subject to:}\; \;  & \int_{\mathcal{X}}w\, d\pi = 0,
\end{split}
\end{equation}
for some $w \in W \subset C_{L}(\mu)$ and $\Pi(\mu)$ being again the set of measures on $\mathcal{P}(\mathcal{X}),$ which are candidate transport plans. Taking into account the constraint, we can write problem (\ref{sup pi}) more concisely:
\begin{equation} \label{sup pi short}
\begin{split}
\phi(f) = \sup_{\pi \in \mathcal{Q}} & \int_{\mathcal{X}}f\, d\pi, \\
\end{split}
\end{equation}
where $\mathcal{Q} = \Pi_{W}(\mu)$ is defined in (\ref{Pi_w}). Without loss of generality, here we consider only one constraint and set this constraint as zero. This is done only for simplifying the presentation, but it can be easily generalised to the case of finitely many constraints $w_1, w_2, \dots, w_n$ taking non-zero values (cf. (\ref{constrainedOTgen})). 

The above problem admits the following dual representation:
\begin{equation} \label{inf h}
\phi(f)  = \inf_{\substack{h + \lambda w \geq f \\ h \in \mathcal{H}, \, \lambda \in \mathbb{R} }}\int_{\mathcal{X}} h\, d\mu , 
\end{equation}
where $ \mathcal{H} \subset C_{L}(\mu),\; h(x_1, \dots, x_n) := \sum_{i=1}^{n}h_{i}(x_i)$ and $\lambda$ denotes a Lagrange multiplier.
The goal is to regularise $\phi(f)$ by penalizing the constraint $h + \lambda w \geq f$:
\begin{equation*}
\begin{split}
\phi_{\theta, \gamma}(f) := \inf_{h \in \mathcal{H}, \, \lambda \in \mathbb{R} }\left\{ \int_{\mathcal{X}} h \, d\mu + \int_{\mathcal{X}} \beta_{\gamma}(f-h-\lambda w)\, d\theta \right\},
\end{split}
\end{equation*}
for a sampling measure $\theta \in \mathcal{P}(\mathcal{X})$, and $\beta_{\gamma}(x) := \dfrac{1}{\gamma}\beta(\gamma x)$ a penalty function, which is parametrised by $\gamma > 0$. We assume that $\beta : \mathbb{R} \to \mathbb{R}_{+}$ is a differentiable increasing convex function. 
\begin{remark}
	An alternative approach here relies on introducing an entropy penalization:
	\begin{equation} \label{entropic_relaxation}
	\phi^{\varepsilon}(f) := \sup_{\pi \in \Pi_{W}(\mu)} \left\{  \int_{\mathcal{X}} f \, d\pi -\varepsilon\int_{\mathcal{X}} \left( \ln\left(\dfrac{d\pi}{d\theta}\right)  -1  \right) d\pi \right\}, 
	\end{equation}
	where $\theta \in \mathcal{P}(\mathcal{X})$ is a prior probability measure and $\varepsilon$ is a positive parameter such that $\lim_{\varepsilon \to \infty}\phi^{\varepsilon}(f) = \phi(f)$. By using the Fenchel-Rockafellar duality theorem, the problem in (\ref{entropic_relaxation}) can be dualised into the following strictly convex optimization problem:
	\begin{equation*}
	\phi^{\varepsilon}(f) = \inf_{h_{i} \in C_{L}(\mu_{i})}\left\{  \sum_{i=1}^{n} \int_{\mathcal{X}_{i}} h_{i} \, d\mu_{i}  + \varepsilon \int_{\mathcal{X}}e\,^{\frac{1}{\varepsilon}(f-h-\lambda w)} \, d\theta   \right\},
	\end{equation*}
	with $h(x_1, \dots, x_n) := \sum_{i=1}^{n}h_{i}(x_i)$.
	For details, see for instance Cuturi \cite{cuturi2013sinkhorn}. Of course, by setting $\varepsilon = 1/\gamma$ and $\beta(x) = \exp(x)$, it can be noticed that the entropy penalization is indeed a special case of the regularization method introduced above.
\end{remark}
In the following result, we show the dual representation of the regularised functional $\phi_{\theta, \gamma}$ and its convergence to $\phi$. 
\begin{theorem} \label{th1NN}
Let $f\in C_{L}(\mu)$ and $w \in W$. Suppose there exists $\pi \in \mathcal{Q}$ such that $\pi \ll \theta$ and $\int_\mathcal{X}\beta^{*}\left( \dfrac{d\pi}{d\theta} \right) d\theta < \infty$, where $\beta^{*}(y) := \sup_{x\in \mathbb{R}}\{ xy - \beta(x)\}$ for all $y \in \mathbb{R}_{+}$ is the convex conjugate of $\beta$. Then
\begin{equation} \label{phi_theta,gamma}
\phi_{\theta, \gamma}(f) = \sup_{ \mu \in \mathcal{Q}}\left\{ \int_{\mathcal{X}} f\, d\mu -\int_{\mathcal{X}} \beta_{\gamma}^{*}\left(\dfrac{d\mu}{d\theta}\right)d\theta \right\}.
\end{equation}
Moreover,
\begin{equation*}
\phi_{\theta, \gamma}(f) - \dfrac{\beta(0)}{\gamma} \leq \phi(f) \leq \phi_{\theta, \gamma}(f) + \dfrac{1}{\gamma}\int_{\mathcal{X}} \beta^{*}\left(\dfrac{d\mu_{\varepsilon}}{d\theta}\right)d\theta + \varepsilon,
\end{equation*}
whenever $\mu_{\varepsilon} \in \mathcal{Q}$ is an $\varepsilon$-optimizer of (\ref{sup pi}) such that $\mu_{\varepsilon} \ll \theta$ and \mbox{$\int_{\mathcal{X}} \beta^{*}\left(\dfrac{d\mu_{\varepsilon}}{d\theta}\right)d\theta < \infty$}.
If $\hat{h} \in \mathcal{H}$ is a minimizer of (\ref{inf h}) and $\hat{\lambda}$ the optimal value of the Lagrange multiplier, then $\hat{\mu} \in \mathcal{P}(\mathcal{X})$ defined by
\begin{equation} \label{joint_distribution}
\dfrac{d\hat{\mu}}{d\theta} := \beta'_{\gamma}(f-\hat{h}-\hat{\lambda} w)
\end{equation}
is a maximizer of (\ref{phi_theta,gamma}).
\end{theorem}

Now we consider a different type of approximation of the initial problem. Let us take a sequence $\mathcal{H}^{1} \subseteq \mathcal{H}^{2} \subseteq \cdots $ of subsets of $\mathcal{H}$, and set $\mathcal{H}^{\infty} := \bigcup_{m \in \mathbb{N}} \mathcal{H}^{m}$. For each $m \in \mathbb{N} \cup \{+\infty\}$, we define the approximated functional $\phi^{m}(f)$ by
\begin{equation} \label{approximated superhedging}
\phi^{m}(f) := \inf_{\substack{h \in \mathcal{H}^{m}\\ h + \lambda w\geq f \\ \lambda \in \mathbb{R}}} \int h \, d \mu.
\end{equation}
Intuitively, here we do not internalize the penalization of the inequality $h + \lambda w \geq f$, but rather take a finite dimensional search space $\mathcal{H}^{m}$. In order for the approximation of $\phi(f)$ by $\phi^{m}(f)$ to be possible, we need the following density condition on $\mathcal{H}^{\infty}$. \vspace{0.5cm}

\noindent \textbf{Condition (D):} For every $\varepsilon > 0$ and $\mu \in \mathcal{P}(\mathcal{X}),$ it  holds that
\begin{itemize}
\item for every $ h \in \mathcal{H}$ there exists $h_1 \in \mathcal{H}^{\infty}$ such that $\int \vert h - h_1 \vert \, d\mu \leq \varepsilon$,
\item there exists $ h_2 \in \mathcal{H}^{\infty}$ such that $\mathds{1}_{K^{c}} \leq h_2$ and $\int h_2 \, d\mu \leq \varepsilon$ for some compact $K \subset \mathcal{X}$.
\end{itemize} 
The condition above allows then to obtain the subsequent convergence result.
\begin{proposition} \label{prop1NN}
Assume that $\mathcal{H}^{\infty}$ is a linear space which contains constant functions. Under Condition (D), one has
\begin{equation*}
\lim_{m\to \infty}\phi^{m}(f) = \phi^{\infty}(f) = \phi(f)
\end{equation*}
for all $f \in C_{L}(\mu)$. 
\end{proposition}

Given a sampling measure $\theta$ and a parametrised penalty function $\beta_{\gamma}$, our final relaxation of the problem is
\begin{equation} \label{final_relaxation}
\phi_{\theta, \gamma}^{m}(f) := \inf_{h \in \mathcal{H}^{m}, \, \lambda \in \mathbb{R}}\left\{ \int_{\mathcal{X}} h \, d\mu + \int_{\mathcal{X}} \beta_{\gamma}(f-h-\lambda w)\, d\theta \right\},
\end{equation}
for all $f \in C_{L}(\mu),  w \in W$. As a consequence of the two approximation steps $\phi_{\theta, \gamma}(f) \to \phi(f)$ for $\gamma \to \infty$ in Theorem \ref{th1NN}, and $\phi^{m}(f) \to \phi(f)$ for $m \to \infty$ in Proposition \ref{prop1NN}, we get the final convergence result.
\begin{proposition}
Suppose that $\mathcal{H}^{\infty}$ satisfies Condition (D) and for every $\varepsilon > 0$ there exists an $\varepsilon$-optimizer $\mu_{\varepsilon}$ of (\ref{sup pi}) such that $\mu_{\varepsilon} \ll \theta$ and $\int_{\mathcal{X}} \beta^{*}\left( \dfrac{d\mu_{\varepsilon}}{d\theta} \right) d\theta < +\infty$. Then, for every $f \in C_{L}(\mu)$ one has $\phi^{m}_{\theta, \gamma} \to \phi(f)$ for $\min\{m,\gamma\} \to \infty$.
\end{proposition}
When the reference measure $\theta$ is chosen as a product measure between the marginals, Proposition 2 in Eckstein et al. \cite{eckstein2018robust} ensures the existence of a $\mu_{\varepsilon}$-optimizer even for non regular transport plans. 
\subsection{Neural networks}
Before moving on, let us first recall the definition and some key properties of  feed-forward neural networks:
\begin{definition} \label{def:neuralnetwork_def}
Let $ L,m_{0},m_{1},\dots,m_{L+1} \in \mathbb{N}$, $\varphi : \mathbb{R} \to \mathbb{R}$, $A_{\ell}: \mathbb{R}^{m_{\ell}} \to \mathbb{R}^{m_{\ell+1}}$, for $\ell = 0,\dots,L$. A mapping $ \mathcal{NN}^{\varphi}_{L,m_0,m_{L+1}} : \mathbb{R}^{m_0} \to \mathbb{R}^{m_{L+1}}$ of the form
\begin{equation*}
\mathbb{R}^{m_0} \ni x \mapsto \overbrace{A(x)}^{\textnormal{output layer}} = \underbrace{\varphi \circ A_{L}}_{L\textnormal{-th hidden layer}} \circ \cdots \circ \underbrace{\varphi \circ A_{1}}_{1\textnormal{st hidden layer}}\circ \overbrace{\varphi \circ A_{0}}^{\textnormal{input layer}}
\end{equation*}
is called a feed-forward fully connected neural network. Here, $m_0$ denotes the input dimension, $L$ the number of hidden layers,\footnote{When $L >2$, the neural network is regarded as a \textit{deep} neural network, whereas when $L \leq 2$ as a \textit{shallow} neural network.} $m_{1},\dots,m_{L}$ the number of nodes (or neurons) of the hidden layers, $m_{L+1}$ the output dimension. Also,  $\varphi$ is a nonlinear activation function applied componentwise at each node of the network, and the affine functions  $A_{\ell}$, for $\ell =0,\dots,L$, are given as $A_{\ell} = M^{\ell}x + b^{\ell}$ for some matrix $M^{\ell} \in \mathbb{R}^{m_{\ell+1} \times m_{\ell}}$ and vector $b^{\ell} \in \mathbb{R}^{m_{\ell+1}}$. For any $i = 1,\dots,m_{\ell+1}$, $j = 1,\dots,m_{\ell}$, the element $M_{ij}^{\ell}$ must be interpreted as the weight of the edge connecting node $i$ of layer $\ell$ to node $j$ of layer $\ell+1$. 
\end{definition}
Denote as $\mathfrak{N}^{\varphi}_{L,m_0,m_{L+1}}$ the set of neural networks of the type $\mathcal{NN}^{\varphi}_{L,m_0,m_{L+1}}:\mathbb{R}^{m_0} \to \mathbb{R}^{m_{L+1}}$ (Definition \ref{def:neuralnetwork_def}). The next fundamental result shows that neural networks can approximate functions with arbitrary accuracy (Hornik \cite{hornik1989multilayer, hornik1991approximation}). 
\begin{theorem}[Universal approximation]
Assume $\varphi$ is bounded and non-constant. The following statements hold:
\begin{itemize}
\item For any finite measure $\mu$ on $\mathbb{R}^{m_0}$, $\mathfrak{N}^{\varphi}_{L,m_0,1}$ is dense in $L^{p}(\mathbb{R}^{m_0},\mu)$, for \mbox{$1\leq p < \infty$}.
\item If, in addition, $\varphi$ is continuous, then $\mathfrak{N}^{\varphi}_{L,m_0,1}$ is dense in $C(\mathbb{R}^{m_0})$ for all compact subsets $X$ of $\mathbb{R}^{m_0}$.
\end{itemize}
\end{theorem}
Further results have been derived based on different assumptions on $\varphi$ and different architectures of the network. For instance, see Hornik \cite{hornik1993some} and a more recent paper by B\"{o}lcskei et al. \cite{bolcskei2019optimal} and references therein. 
\paragraph{Modeling $\mathcal{H}^{m}$ via neural networks.} Encouraged by mentioned results, we can move back to our problem $\phi_{\theta, \gamma}^{m}(f)$, which we restate here for readability: 
\begin{equation} \label{final approximated superhedging}
\phi_{\theta, \gamma}^{m}(f) = \inf_{h \in \mathcal{H}^{m}, \, \lambda \in \mathbb{R}}\left\{ \int_{\mathcal{X}} h \, d\mu + \int_{\mathcal{X}} \beta_{\gamma}(f-h - \lambda w)\, d\theta \right\},
\end{equation}
for all $f \in C_{L}(\mu), w \in W$. 

In the following, we  work with networks with fixed number of layers and nodes, but unknown parameter values (weights). As common practice, we  set an equal number $m$ of nodes for each hidden layer, that is, $ m= m_{1} = \dots = m_{L}$. Also, we denote by $\mathfrak{N}^{\varphi}_{L,n_0,m_{L+1}}(\Theta_{m})$ the sets of neural networks with parameters $\Theta_{m} \subset \mathbb{R}^{d_{m}}$, for some $d_{m} \in \mathbb{N}$ depending on the architecture of the network.

Now, recall that the elements in $\mathcal{H}^{m} \subseteq\mathcal{H}$ are functions in $C(\mathcal{X})$. For this reason, we can approximate them by means of neural networks.\footnote{It would be natural to wonder in which situations Condition (D) is satisfied in the context of neural networks. This point is discussed quite in detail in Lemma 3.3 of Eckstein and Kupper \cite{eckstein2018computation}. It turns out that it suffices to have standard smoothness properties of the activation function $\varphi$ and at least one hidden layer in the architecture of the network.} The minimization problem in (\ref{final approximated superhedging}) is reduced to  finding the optimal parameters $\Theta_{m}$ for a set of neural networks:
\begin{equation*} 
\phi_{\theta, \gamma}^{m}(f) = \inf_{h \in \mathfrak{N}^{\varphi}_{L,m_0,m_{L+1}(\Theta_{m}), \, \lambda \in \mathbb{R}}}\left\{ \int_{\mathcal{X}} h \, d\mu + \int_{\mathcal{X}} \beta_{\gamma}(f-h-\lambda w)\, d\theta \right\}.
\end{equation*}
This last formulation is a finite dimensional problem and can be solved using (stochastic) gradient descent. For our applications, we will alternate the Adam optimizer (Kingma and Ba \cite{kingma2014adam}), which is an extension to stochastic gradient descent, with few iterations of gradient descent for optimizing only with respect to the Lagrange multipliers (associated with the constraints of the problem). 

In theory, as it is most often the case with neural networks, the objective function ends up being non-convex and the algorithm may only converge to a local minimum/maximum. However, the effectiveness of stochastic gradient descent (and its versions) has been extensively showcased in practice also for non-convex problems, making it one of the most widely used algorithm for deep learning.

\begin{algorithm}[!h]
\caption{Linearly constrained OT via neural nets:  ``alternating'' gradient descent optimization}
\begin{algorithmic}
	\STATE \textbf{Inputs:} marginal distributions $\mu_{1}, \mu_{2}$; reference measure $\theta$; target function $f$; constraint $w$; batch size $p$; penalty function $\beta_{\gamma}$; number of gradient descent iterations $n_{gc}$; hyper-parameters for the neural networks architecture $\Theta_{m}$. \\
	\textbf{Require:} random initialization of weights $\textbf{w}_{1}, \textbf{w}_{2}$; initialization of Lagrange multiplier $\lambda$.
	\WHILE{not converged}
	\STATE sample $\{x_i\}_{i=1}^{p} \sim \mu_{1}$;
	\STATE sample $\{y_i\}_{i=1}^{p} \sim \mu_{2}$;
	\STATE approximate  $h_1(\textbf{w}_{1}; x_1, x_2, \dots, x_p), \, h_2(\textbf{w}_{2}; y_1, y_2, \dots, y_p)$ via NNs;
	\STATE evaluate $\phi_{\theta,\gamma}^{m}(\textbf{w}_{1}, \textbf{w}_{2}, \lambda; f) = \int h_1 d\mu_1 + \int h_2 d\mu_2 + \int \beta_{\gamma}(f-h_1 - h_2-\lambda w) d\theta$;
	\STATE $ \textbf{w}_{1}, \textbf{w}_{2} \leftarrow \mbox{Adam}(\phi_{\theta,\gamma}^{m}(\textbf{w}_{1}, \textbf{w}_{2}, \lambda; f))$;
	\FOR{\_ in range $n_{gc}$}
	\STATE sample $\{x_i\}_{i=1}^{p} \sim \mu_{1}$;
	\STATE sample $\{y_i\}_{i=1}^{p} \sim \mu_{2}$;
	\STATE $\lambda \leftarrow \mbox{GradientDescent}(\phi_{\theta,\gamma}^{m}(\textbf{w}_{1},\textbf{w}_{2}, \lambda; f));$
	\ENDFOR
	\ENDWHILE
	\end{algorithmic}
\label{OT via NN: AGD}
\end{algorithm}

In the box above (Algorithm \ref{OT via NN: AGD}), we provide the pseudocode for the algorithm. For simplicity, therein we consider an optimal transport problem with only two marginals and one constraint. In a nutshell, the logic is the following: first, we sample batches from the marginal distributions; then, we approximate $h$ in (\ref{final approximated superhedging}) via suitable neural networks and optimise (maximise or minimise) the objective functional with respect to the weights of the neural networks. For each of these iterations, we perform a  few rounds of gradient descent in order to optimise with respect to the Lagrange multiplier. Then, if a stopping rule cannot be applied, the algorithm goes back to sampling and repeats the procedure.\footnote{Notice that, even for higher dimensions and more constraints, the logic of the algorithm stays exactly the same.}

\section{Arbitrage bounds of derivatives prices} \label{applications}
	In this section, we perform some numerical experiments as applications of the methodology presented above. All computations are performed in Python using TensorFlow. Concerning the neural networks architecture, we work with networks with 2 hidden layers, 128 nodes per hidden layer and rectified linear unit (ReLU) activation function. Also, we use an $L^{2}$ penalization function $\beta$ with $\gamma = 100$. The networks are trained with a batch of $2^{8}$ draws for 20,000 iterations, for each of which we perform 10 iterations of gradient descent, as explained in the previous section. For completeness, in Appendix \ref{further_numerics} we provide a sensitivity analysis with respect to the choice of $\gamma$ and the hyperparameters.

We develop two specific applications. In the first one, we  focus on basket options given related option prices.  In the second one, we provide bounds on Asian basket options. For a general framing of the problems that we are going to consider, see Section \ref{setting}. All examples are given for illustration purpose only and do not reflect  empirical distributional properties of asset prices and/or market data. Furthermore, for the ease of presentation, we assume a 0\% risk-free rate.

 Thus, let us consider a market with six assets $S_i, \, i = 1, \dots, 6$.  We assume that $S_{1} \sim \mathcal{N}(0,1)$, $S_{2} \sim \mathcal{N}(0,2)$, $S_{3} \sim \mathcal{N}(1,2)$, $S_{4} \sim \mathcal{N}(2,1),S_{5} \sim \mathcal{N}(1,1) $ and $S_{6} \sim \mathcal{N}(2,2)$, where $\mathcal{N}(\mu,\sigma)$ indicates a Gaussian probability distribution with mean $\mu$ and standard deviation $\sigma$.  Choosing normal distributions for assets' prices is not realistic but ensures that special limit cases can be studied in explicit form. Such results are useful to check the effectiveness of our approach.

\subsection{Basket options with additional information}

\subsubsection{4-basket option given the price of a 2-spread}

Let us start with computing bounds on a basket option written on four assets $S_{1},\, S_{2},\, S_{3}$ and $S_{4}$, given information on a spread option on two assets of the basket.  Denoting by $K$ the strike price, the price of the 4-basket option is then given by
\begin{equation*}
p_{4D} := \mathbb{E}\left[ \left(S_1 + S_2 + S_3 + S_4 - K \right)^{+}\right].
\end{equation*}
Also, the price of the spread option between $S_{4}$ and $S_{3}$ is known and given by
\begin{equation*}
p_{1} = \mathbb{E}\left[ \left(S_4 - S_3 \right)^{+}\right].
\end{equation*}
Observe that $p_{1}$ must approximately be within an interval $p_{1} \in [1.08,\, 1.76]$. Specifically, this interval corresponds to the lower and upper bounds on the price of a spread option on $S_4-S_3$ that are obtained with an antimonotonic and a comonotonic copula, respectively. For normal distributions, such bounds can be computed explicitly using the following expression:
\begin{equation}\label{e1}
\mathbb{E}\left[(X-a)^{+}\right]=(\mu-a)\Phi\left(\frac{\mu-a}{\sigma}\right)+\sigma\varphi\left(\frac{\mu-a}{\sigma}\right),
\end{equation}
where $\Phi$ denotes the cumulative density function (cdf) of a standard $\mathcal{N}(0,1)$ and $\varphi$ its probability density function (pdf). For $p_1$, the lower bound is obtained  when $S_{4}$ and $S_{3}$ are antimonotonic and it is then equal to $\mathbb{E}\left[(1-Z)^{+}\right] = \Phi(1)+\varphi(1)\approx 1.08$, where $Z\sim\mathcal{N}(0,1)$. The upper bound is equal to $\mathbb{E}\left[(3Z+1)^{+}\right] = \Phi(1/3)+3\varphi(1/3)\approx1.76.$ 

For an arbitrary  range of strikes $K \in \left[0, \, 0.7\right]$ and for a feasible price $p_{1} \in [1.08, \, 1.76],$ we display the upper and lower bounds for the price $p_{4D}$ of the basket option  obtained using our methodology in \mbox{Figure \ref{fig:4Db-1s}}. 
\begin{figure}[!h]
	\centering
\includegraphics[width=0.8\textwidth,scale = 0.6, height=8cm]{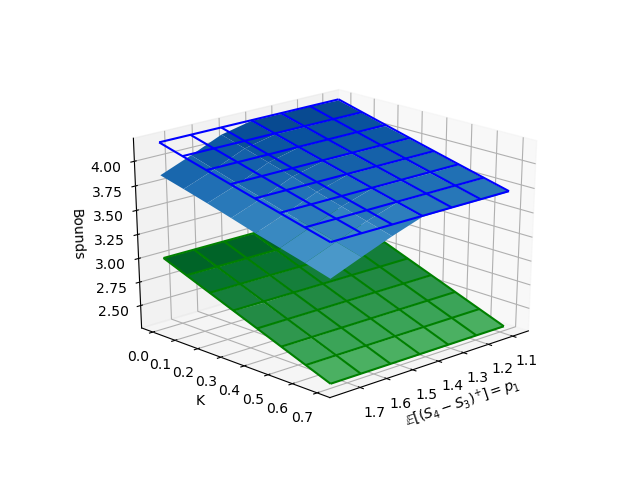} 
	\caption{ Bounds on 4-basket option with strike $K \in [0, 0.7]$, given the price of a 2-spread option.} \label{fig:4Db-1s}
\end{figure}

 Notice that the grids correspond to the unconstrained bounds (e.g., bounds computed without incorporating additional information to the information on the marginal distributions). The upper unconstrained bound is obtained by the comonotonic structure, which  in our example with normal distribution has an explicit formula as a function of $K$:
  \begin{equation*}
  \label{UBU}\mathbb{E}\left[(S_{1}^c+S_{2}^c+S_{3}^c+S_{4}^c-K)^{+}\right]=(3-K)\Phi\left(\frac{3-K}{6}\right)+6\varphi\left(\frac{3-K}{6}\right).
  \end{equation*}
 The lower unconstrained bound is obtained by deriving a dependence structure that leads to a minimum in convex order for the sum. In general, one can obtain an approximate solution by applying the RA (Rearrangement Algorithm of Puccetti and R\"uschendorf \cite{puccetti2012computation}). In the situation of this example, there exists a dependence such that the sum of the four assets is identically constant and equal to 3. Thus the minimum unconstrained bound is simply equal to $\label{LBU}(3-K)^{+}$.
 
 From Figure \ref{fig:4Db-1s}, we observe that the lower bound does not change if we include information on the spread between two assets of the basket. This is not surprising as it is possible to achieve complete mixability among the four assets by assuming that $S_{3}$ and $S_{4}$ are antimonotonic. Namely, let $Z\sim \mathcal{N}(0,1)$ and take $S_{4}=2+Z, \, S_{3}=1-2Z,\, S_{1}=-Z$ and $S_{2}=2Z$. On the other hand, the upper bound becomes lower when the price of the spread option becomes higher. Intuitively, an unconstrained upper bound on a basket option simply requires that all assets are in comonotonic order; however, a higher price (close to the upper bound) of a spread option  between two assets implies that these two assets are in antimonotonic order, reducing the (overall) maximum possible degree of comonotonicity. 

\subsubsection{4-basket option given the prices of two 2-baskets}

We now consider the same 4-basket option, given the prices of two basket options written on $S_{1}$ and $S_{2}$, and $S_{3}$ and $S_{4}$, respectively. We set again the strike $K = 0.1$.  
The prices of the baskets on $S_{1} + S_{2}$ and $S_{3}+S_{4}$ are given by
\begin{equation*}
\begin{split}
p_2 = \mathbb{E}\left[ \left(S_1 + S_2 \right)^{+}\right], \\
p_3 = \mathbb{E}\left[ \left(S_3 + S_4 \right)^{+}\right].
\end{split}
\end{equation*}
Explicit bounds for these prices are obtained using \eqref{e1}.  Specifically, for $p_2 \in [\varphi(0), \, 3\varphi(0)]\approx[0.40, \, 1.20]$ and $p_3 \in [3\Phi(3)+\varphi(3), \, 3\Phi(1)+3\varphi(1)]\approx [3.00, \, 3.25] $, we display upper and lower bounds for the 4-basket option price as in Figure \ref{fig:4Db-2b}. 

\begin{figure}[!h]
	\centering
	\includegraphics[width=0.7\textwidth,height=8cm]{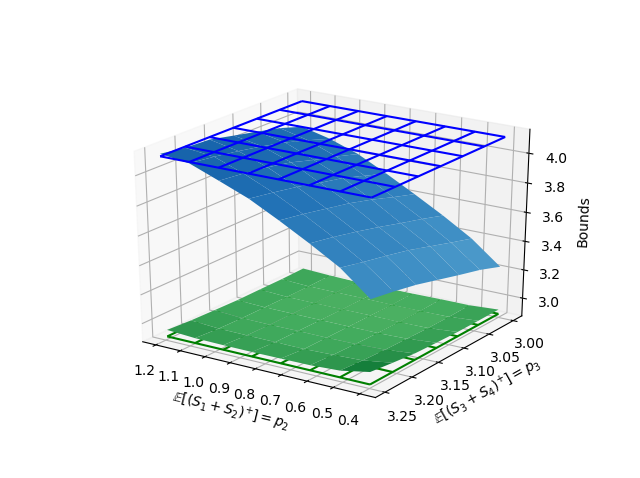} 
	\caption{ Bounds on a 4-basket option on $S_{1},\, S_{2},\, S_{3}, S_{4}$ with strike $K \in [0, 0.7]$, given the price of a 2-spread option on $S_{3}, S_{4}$. } \label{fig:4Db-2b}
\end{figure}

Note that we can similarly deal with the case of overlapping constraints, e.g., the case when the prices of two baskets on $S_{1}$ and $S_{2}$, and $S_{2}$ and $S_{4}$, respectively, are known. 



\subsubsection{5-basket option given the prices of a 3-basket and a 2-spread}
Let us increase slightly the dimensionality of the problem, and consider a basket option on five assets $S_1, S_2, S_3, S_4, S_5$ with strike $K$. As before, its price is given by  
\begin{equation*}
p_{5D} := \mathbb{E}\left[ \left(S_1 + S_2 + S_3 + S_4 + S_5 - K \right)^{+}\right].
\end{equation*}
Additionally, we set the prices of a 3-basket option on $S_{1} + S_{2} + S_{5}$ and a 2-spread option on $S_{4} - S_{3}$. These prices are given by 
\begin{equation*}
\begin{split}
& p_5 = \mathbb{E}\left[ \left(S_1 + S_2 + S_5 \right)^{+}\right], \\
& p_1 = \mathbb{E}\left[ \left(S_4 - S_3 \right)^{+}\right].
\end{split}
\end{equation*}
Here we set $p_{1} \in [1.08,\, 1.76]$ (see first example) and $p_5 \in [1,\, \Phi(1/4)+4\varphi(1/4)]\approx[1, 2.14]$. Notice that the lower bound of $p_{5}$ is simply obtained when the three assets are mixable: $S_{1}=Z, \, S_{2}=-2Z$ and $S_{5}=1+Z$, again for some $Z \sim \mathcal{N}(0,1)$; on the other hand, the upper bound corresponds to $\mathbb{E}\left[(1+4Z)^{+}\right]$ for which there is a closed-form expression in \eqref{e1}. 
Thus, still for $K = 0.1$, we obtain upper and lower bounds on the 5-basket option as in Figure \ref{fig:5Db-1s1b}.

\begin{figure}[!h]
		\centering
\includegraphics[width=0.8\textwidth,scale = 0.6]{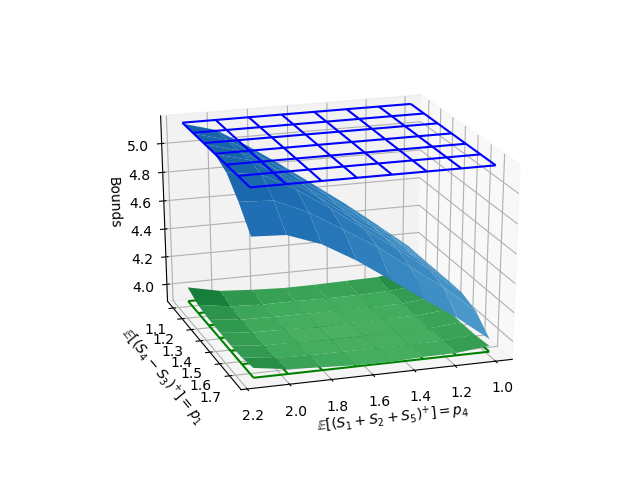} 
	\caption{ Bounds on a 5-basket option on $S_{1},\, S_{2},\, S_{3}, \, S_{4}, \,S_{5}$ with strike $K = 0.1$, given the price of a 3-basket option on $S_{1}, \, S_{2}, \, S_{5} $ and a spread option on $S_{3}, S_{4}$.} \label{fig:5Db-1s1b}
\end{figure}

Also in this case, we display as well the unconstrained bounds, which only account for the information on the marginal distributions of the five assets. These bounds can again be computed  explicitly and we do not need the RA. For the unconstrained  lower bound, a basket option with strike $K$ on the sum of the five assets is always larger than $(4-K)^+$, that is obtained in the very special case in which the dependence is for instance driven by two independent random variables $\mathcal{N}(0,1)$, denoted as $Z_1$ and $Z_2$, by constructing the five assets as follows: $S_{1}=Z_1$, $S_{2}=2Z_1$, $S_{3}=1-2Z_1$, $S_{4}=2-\frac{1}{2}Z_1+\frac{\sqrt{3}}{2} Z_2$, and  $S_{5}=2-\frac{1}{2}Z_1-\frac{\sqrt{3}}{2}Z_2$. We can then simply compute the lower bound when $K=0.1$, that is 3.9. For the upper bound it is also straightforward to obtain an explicit expression by using \eqref{e1}:
$$ \mathbb{E}\left[(4+7Z-K)^{+}\right]=(4-K)\Phi\left(\frac{4-K}{7}\right)+7\varphi\left(\frac{4-K}{7}\right).$$
To obtain the unconstrained upper bound in Figure  \ref{fig:5Db-1s1b}, one can simply replace $K$ by its value 0.1.

In Figure  \ref{fig:5Db-1s1b},  we notice that the upper bound is most tightened when the price of the basket on $S_{1} + S_{2} + S_{5}$ is close to its lower bound and the price of the spread on $S_{4} - S_{3}$ is close to its upper bound (which happens when the two assets are in antimonotonic order).

In fact, the bounds on the 5-basket when the basket and the spread options are both equal to their respective upper bounds can be computed explicitly. In this case,
$$\mathbb{E}\left[(4+3Z-K)^+\right]\leqslant \, p_{5D} \, \leqslant \mathbb{E}\left[(4+5Z-K)^+\right].$$

Similarly, when the basket is at the upper bound and the spread is at its lower bound, we have
$$\mathbb{E}\left[(4+Z-K)^{+}\right]  \leqslant \, p_{5D} \leqslant \, \mathbb{E}\left[(4+7Z-K)^+\right].$$

All these bounds are explicit and can be computed by means of \eqref{e1}. We have checked that they are approximately consistent with the output of the algorithm.

\subsubsection{5-basket option given the prices of two 2-spreads}
As a final example for this section, let us consider a similar 5-basket option on $S_{1},\, S_{2},\, S_{3}, \, S_{4}, \,S_{6}$,  given prices of a 2-spread option on $S_{2} - S_{1}$ and a 2-spread option on $S_{4} - S_{3}$. These prices are given by 
\begin{equation*}
\begin{split}
&p_6 = \mathbb{E}\left[ \left(S_2 - S_1\right)^{+}\right], \\
&p_1 = \mathbb{E}\left[ \left(S_4 - S_3 \right)^{+}\right].
\end{split}
\end{equation*}
For  $p_{1} \in [1.08,1.76]$ and $p_6 \in [\varphi(0), 3\varphi(0)]=[\frac{1}{\sqrt{2\pi}}, \frac{3}{\sqrt{2\pi}}] \approx[0.40, 1.20] $, we obtain upper and lower bounds as in \mbox{Figure \ref{fig:5Db-2s}}.
\begin{figure}[!h]
		\centering
\includegraphics[width=0.8\textwidth,scale = 0.6]{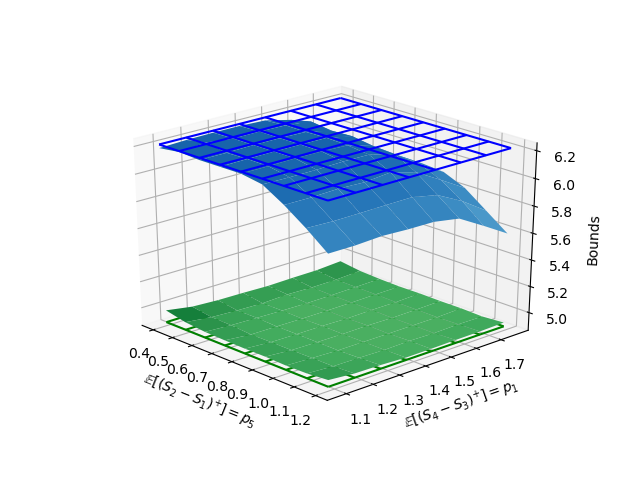} 
	\caption{  Bounds on a 5-basket option on $S_{1},\, S_{2},\, S_{3}, \, S_{4}, \,S_{6}$ with strike $K = 0.1$, given the prices of two 2-spread option on $S_{1}, \, S_{2}$ and $S_{3}, S_{4}$ respectively.} \label{fig:5Db-2s}
\end{figure}
In this case, the constrained upper bound appears to be tightened when the spread options approach their upper bounds and almost exactly as the unconstrained upper bound when the spread options both move towards their lower bounds.

\newpage
\subsection{Asian-style basket options \label{PaDe}}
We end the numerical section by an illustration of the methodology on a weakly path-dependent option, e.g., a payoff that depends on underlying assets at a finite number of past dates. The additional difficulty is to include the martingale condition, that is, to ensure that the marginal distribution of the asset at each intermediary date satisfies a martingale condition.   Specifically, we illustrate the study with an Asian-style option in which the underlying is the sum of the basket price at some intermediary date and the basket price at maturity of the option. By dividing this payoff by two, one would get an option on the arithmetic average of the basket, and thus an Asian option.

\subsubsection{Asian (2+2)-basket option given the price of a 2-spread}
Let us assume a two-period financial market and consider two assets $S_1$ and $S_2$. Denoting by $S_{i,T}$ the distributions of asset $i=1,2$ at time $T=1$ and $T=2$, respectively, we set the following: $S_{1,1} \sim \mathcal{N}(0,1), \, S_{1,2} \sim \mathcal{N}(0,2)$ and $S_{2,1} \sim \mathcal{N}(1,1), \, S_{2,2} \sim \mathcal{N}(1,2)$. Also, we fix $K=2$. We are interested in deriving bounds on the Asian (2+2)-basket option given by 
\begin{equation*}
p_{2D}^{Asian} = \mathbb{E}\left[\left( S_{1,1} + S_{1,2} + S_{2,1} + S_{2,2} - K \right)^{+} \right].
\end{equation*}
We include information on the price of a 2-spread option on $S_1$ and $S_2$ at the intermediary date $T=1$, i.e., the underlying is  $S_{2,1} - S_{1,1}$:
\begin{equation*}
p_7 = \mathbb{E}\left[\left( S_{2,1} - S_{1,1}  \right)^{+} \right].
\end{equation*}
Due to standard financial arguments, $S_{1}$ and $S_{2}$ need to be martingales, so the marginal distributions at time $T=1,2$ are set to be in increasing convex order. Recall that, by Strassen \cite{strassen1965existence}, there exists a martingale $(M_1,M_2)$ such that $M_1 \sim \mu_{1}$ and $M_2 \sim \mu_{2}$  if and only if measures $\mu_1$ and $\mu_{2}$ are in convex order. This adds a further constraint to the problem, as of course now the optimal coupling needs to satisfy these martingality conditions. In practice, such condition is imposed by applying Lemma 2.3 in Beiglbock et al. \cite{beiglbock2013model}, which we report here for completeness:
\begin{lem}\label{LEM2}
Let $\pi \in \Pi(\mu_{1}, \dots, \mu_{n})$. Then the following statements are equivalent: 
\begin{itemize}
\item[1.] $\pi \in \mathcal{M}(\mu_{1}, \dots, \mu_{n})$, where $\mathcal{M}(\mu_{1}, \dots, \mu_{n})$ denotes the set of martingales with marginal distributions $\mu_{1}, \dots, \mu_{n}$;
\item[2.] For every $i \leq j \leq n-1$ and for every continuous bounded function $\Delta : \mathbb{R}^{j} \to \mathbb{R}$, we have
\begin{equation}\label{penalty_martingale}
\int_{\mathbb{R}^{n}}\Delta(x_{1},\dots,x_{j})(x_{j+1}-x_{j})d\pi(x_{1},\dots,x_{n})=0.
\end{equation}
\end{itemize}  
\end{lem}
Notice that now $\Delta(x_{1},\dots,x_{j})$ can also be approximated via neural networks and (\ref{penalty_martingale}) can be enforced by penalization, as shown previously.  In our case, $n=2$ and thus Lemma \ref{LEM2} simplifies into one condition that for every bounded function $\Delta : \mathbb{R}\to \mathbb{R}$, we have
\begin{equation*}\label{penalty_martingaleSpecial}
\int_{\mathbb{R}^{2}}\Delta(x_{1},x_{2})(x_{2}-x_{1})d\pi(x_{1},x_{2})=0.
\end{equation*}
We then approximate $\Delta$ via neural networks and the previous methodology applies.

\begin{figure}[!h]
	\centering
	\includegraphics[width=1\textwidth,height=6cm]{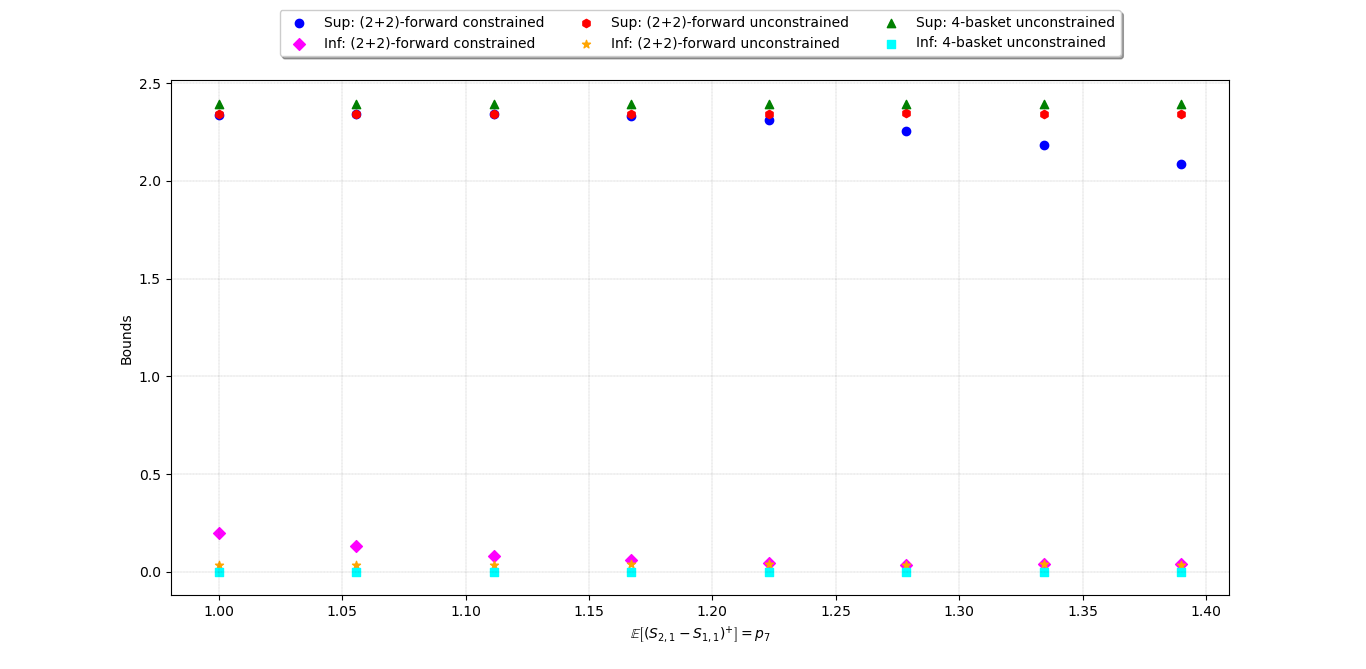} 
	\caption{ Bounds on a Asian-style (2+2)-basket option on $S_{1,1}, \, S_{1,2}, S_{2,1}, \, S_{2,2}$ with strike ${K = 0.1}$, given the price of a 2-spread option on $S_{1,1}$ and $S_{2,1}$.} 
	\label{fig:forward}
\end{figure}

Thus, for $p_7 \in \left[1, \Phi(1/2) + 2\phi(1/2) \right] \approx \left[1, \, 1.39 \right] $, we obtain upper and lower bounds for the Asian (2+2)-basket option as in Figure \ref{fig:forward}.

The red circles represent the upper bounds on the Asian basket option. The green empty circles represent the unconstrained upper bounds (assuming no further information on $p_7$ and no martingale constraints (it can be thought as a simple 4-basket)), the red circles are the upper bounds without martingale constraints but assuming knowledge on $p_7$ (constrained 4-basket). We reach similar conclusions  for the lower bounds. As expected, upper bounds from the constrained, martingale OT problem are below (even if just slightly) the ones from the constrained OT problem. 

Additional examples of bounds on path-dependent derivatives can be found in Eckstein et al. \cite{eckstein2019robust}. 

\section*{Conclusive remarks} In this paper, we show how the approach by Eckstein and Kupper \cite{eckstein2018computation} can also be applied in a very flexible and effective way to  problems related to the computation of bounds on multi-asset derivatives when further information on the assets is known, e.g., via prices of derivatives already traded in the market. In most cases, the computation of such bounds cannot be done explicitly and other algorithms (as, for instance, the Rearrangement Algorithm) do not allow to incorporate this type of constraints. 

Our experiments show that, by means of Algorithm \ref{OT via NN: AGD}, we can obtain reliable price bounds in many different settings.  However, we point out that, due to the approximation of the dual objective in (\ref{final_relaxation}), the actual joint distribution achieving such bounds can only be obtained in an approximated form via (\ref{joint_distribution}) (see De Gennaro Aquino and Eckstein \cite{aquino2020minmax} for a study on primal-dual methods for optimal transport that can be used to address this issue).

On a further note, since our framework allows to detect arbitrage opportunities in the multi-asset derivatives market, we expect that more practical applications can be developed by using all the information available in the market at some given point in time. We leave as an avenue for future research the investigation of how such arbitrage opportunities can be exploited.

\newpage 
\bibliographystyle{siam}
\bibliography{Bounds_NN-references}

\newpage
\appendix
\begin{appendices}
\section{Duality results} \label{appendix_Kantorovic}
Kantorovich duality has been stated and proved in several settings. Here, following again Zaev \cite{zaev2015monge}, we report details and preliminary results for the proof of the case with additional linear constraints examined in Section \ref{MK with constraints}.

\subsection{Well-definiteness of the seminorm $\Vert \cdot \Vert_{L}$} \label{well-definiteness}
Firstly, we need to prove that the seminorm $\Vert \cdot \Vert_{L}$ in (\ref{seminorm}) is well-defined, as by the following definition: 
\begin{definition}
	A seminorm on a generic vector space $V$ is a function $\Vert \cdot \Vert : V \to \mathbb{R}_{+}$ such that the following properties are satisfied: for all $v, w \in V$, and any scalar $a$,
	\begin{itemize}
		\item $\Vert v \Vert \geq 0$ (non-negativity),
		\item $\Vert a \, v \Vert = \vert a\vert \, \Vert v \Vert $ (absolute homogeneity), and
		\item $\Vert v + w \Vert \leq \Vert v  \Vert + \Vert w \Vert$ (subadditivity).
	\end{itemize} 
\end{definition} 
The following proposition then holds.
\begin{proposition} 
	$\Vert \cdot \Vert_{L}$ is a well-defined seminorm on $C_{L}(\mu)$.
\end{proposition}
\begin{proof}
	It is easy to check that $\Vert \cdot \Vert_{L}$ is finite for all $c \in C_{L}(\mu)$, as
	\begin{equation*}
	\sup_{\pi \in \Pi(\mu)}\int_{\mathcal{X}} \vert c \vert d\pi \leq \sum_{i=1}^{n}\int_{\mathcal{X}_{i}}h_{i}\, d\mu_{i} < \infty,
	\end{equation*}
	non-negative and absolutely homogeneous. It also holds that, for any $c,g \in C_{L}(\mu)$,
	\begin{equation*}
	\Vert c + g \Vert_{L} = \sup_{\pi \in \Pi(\mu)}\int_{\mathcal{X}}  \vert  c +g   \vert d\pi \leq \sup_{\pi \in \Pi(\mu)}\int_{\mathcal{X}} \vert  c    \vert d\pi + \sup_{\pi \in \Pi(\mu)}\int_{\mathcal{X}} \vert  g   \vert d\pi = \Vert c \Vert_{L} + \Vert g \Vert_{L},
	\end{equation*}
	proving that $\Vert \cdot \Vert_{L}$ is subadditive.
\end{proof}

\subsection{Compactness and continuity} \label{compactness-continuity}
To prove compactness of $\Pi_{W}(\mu)$, we recall the following result by Prokhorov \cite{prokhorov1956convergence}:
\begin{lem}
Let $\mathcal{X}$ be a Polish space. A set $\mathcal{P} \in \mathcal{P}(\mathcal{X})$ is precompact\footnote{A subset $S$ of a metric space is called \textit{precompact} or \textit{relatively compact} if its closure $\bar{S}$ is compact.} in the weak topology if and only if it is tight, that is, for every $\varepsilon > 0$ there exists a compact subset $K_{\varepsilon}$ such that $\mu[\mathcal{X} \backslash K_{\varepsilon}] \leq \varepsilon$ for all $\mu\in \mathcal{P}$.
\end{lem}  
By Lemma 4.4 in Villani \cite{villani2008optimal}, we know that $\Pi(\mu)$ is tight in $\mathcal{P}(\mathcal{X})$ and then, by Prokhorov's theorem, that it is relatively compact in the topology of weak convergence. In order to show that $\Pi_{W}(\mu)$ is also compact in such topology, it takes to show that is closed. By definition,  $\Pi_{W}(\mu)$ is the intersection of the sets $\left\{ \pi \in \Pi(\mu): \int_{\mathcal{X}}w \, d\pi = 0  \right\}$, for some function $w \in W \subset C_{L}(\mu)$. By Lemma 4.3 in Villani \cite{villani2008optimal}, each of these sets is closed, which implies that $\Pi_{W}(\mu)$ is also closed.

Now, let $C_{b}(\mathcal{X})$ denote the space of bounded continuous functions on $\mathcal{X}$. To prove lower semi-continuity of the functional $\pi \to \int_{\mathcal{X}} c \, d\pi$, we need to check that for any sequence of transport plans $\{\pi_{k}\}$ such that $\lim_{k}\int_{\mathcal{X}} \rho d\pi_{k} = \int_{\mathcal{X}} \rho d\pi$ for any $\rho \in C_{b}(\mathcal{X})$, we have $\lim_{k} \int_{\mathcal{X}} c d\pi_{k} = \int_{\mathcal{X}} c d\pi$. Before doing that, we need the following lemma, which shows that $C_{b}(\mathcal{X})$ is dense in $C_{L}(\mu)$.
\begin{lem}
	$C_b (\mathcal{X})$ is dense in $C_{L}(\mu)$ with respect to the seminorm  $\Vert \cdot \Vert_{L}$.
\end{lem}
\begin{proof}
	Pick $g \in C_{L}(\mu)$, and let $\vert g \vert \leq h \in \mathcal{H}$. Let $k \in \mathbb{N}$, $g^{k} := \min\{k,g \}$ and $g^{k}_{k} := \max\{ g^{k},-k \}, g^{k}_{k} \in C_{b}(\mathcal{X})$. We need to show that $\Vert g - g^{k}_{k} \Vert_{L} \to 0$ as $k \to \infty$. First, notice that
	\begin{equation*}
	\Vert g - g^{k}_{k} \Vert_{L} \leq  \Vert g - g^{k} \Vert_{L} +  \Vert g^{k} - g^{k}_{k} \Vert_{L}.
	\end{equation*}
	Then, by Lebesgue dominated convergence theorem,
	\begin{equation*}
	\Vert g - g^{k} \Vert_{L} = \sup_{\pi \in \Pi(\mu)} \int_{\mathcal{X}} \max\{g-k,0\}d\pi \leq \sum_{i=1}^{n}\int_{\mathcal{X}_i} \max\left\{h_{i}-\dfrac{k}{n},0\right\} \, d\mu_{i} \to 0, \quad \mbox{as } k \to \infty,
	\end{equation*}
	\begin{equation*}
	\Vert g^{k} - g_{k}^{k} \Vert_{L} = \sup_{\pi \in \Pi(\mu)} \int_{\mathcal{X}} \max\{-g^{k}-k,0\}d\pi \leq \sum_{i=1}^{n} \int_{\mathcal{X}_i} \max\left\{h_{i}-\dfrac{k}{n},0\right\} \, d\mu_{i} \to 0, \quad \mbox{as } k \to \infty.
	\end{equation*} 
\end{proof}
Since $C_{b}(\mathcal{X})$ is dense in $C_{L}(\mu)$, there exists a sequence  $\{\rho_{n}\}_{n\in \mathbb{N}}$ of continuous functions converging uniformly to $c$  in the $\Vert \cdot \Vert_{L}$-topology. Also, by existence of the limits $\lim_{k} \int_{\mathcal{X}} \rho_{N}d\pi_{k}$ and $\lim_{n}\int_{\mathcal{X}} \rho_{n}d\pi_{K}$ for sufficiently large $N$ and $K$, we have that
\begin{equation*}
\lim_{k} \int_{\mathcal{X}} c \,d\pi_{k} = \lim_{k} \lim_{n} \int_{\mathcal{X}} \rho_{n}d\pi_{k} =\lim_{n} \lim_{k} \int_{\mathcal{X}} \rho_{n}d\pi_{k} =\lim_{n} \int_{\mathcal{X}} \rho_{n}d\pi = \int_{\mathcal{X}} c \,d\pi.
\end{equation*}
Thus, if and only if $\Pi_{W}(\mu)$ is non-empty, continuity and compactness guarantee the existence of a solution for the Kantorovich problem with additional linear constraints.

\subsection{Preliminary results for the proof of Theorem \ref{constrained Kantorovich duality}} \label{lemmas-appendix}
In this section, we recall two results, Lemma \ref{Kantorovich duality} and Lemma \ref{minimax} that are the key ingredients in proving Theorem \ref{proof-Kantorovich-constraints}.
	\begin{lem}[Kantorovich duality] \label{Kantorovich duality}
	Let $\mathcal{X}_{1}, \dots,  \mathcal{X}_{n}$ be Polish spaces, $\mathcal{X} = \mathcal{X}_{1} \times \cdots \times  \mathcal{X}_{n}$, $\mu = (\mu_{i} \in \mathcal{P}(\mathcal{X}_{i}))$ for $i = 1, \dots, n$, $f \in C_{L}(\mu)$. Then
	\begin{equation*}
	\inf_{\pi \in \Pi(\mu)}\int_{\mathcal{X}}f\,d\pi =   \sup_{h \leq f, \, h \in \mathcal{H}} \sum_{i=1}^{n}\int_{\mathcal{X}_i} h_{i} \, d\mu_{i}.
	\end{equation*}
\end{lem}

\begin{proof}
	Let $T: \mathcal{H} \to\mathbb{R}$ be a linear functional defined by
	\begin{equation*}
	T(h) = \sum_{i=1}^{n}\int_{\mathcal{X}_{i}} h_{i}d\mu_{i}.
	\end{equation*}
	It is positive and continuous with respect to the seminorm $\Vert \cdot \Vert_{L}$. Let $ U: C_{L}(\mu) \to \mathbb{R}$ be another functional defined by
	\begin{equation} \label{U}
	U(c) = \inf_{h \in \mathcal{H}}\{T(h) : h\geq c\},
	\end{equation}
	where $\mathcal{H} \subset C_{L}(\mu)$. Then, $U(h) = T(h)$ for $h \in \mathcal{H}$.
	In order to apply the Hahn-Banach theorem, it takes to prove that $U$ is subadditive and positively homogeneous. For subadditivity, notice that
	\begin{equation*}
	\begin{split}
	U(c+g)  = & \inf_{h \in \mathcal{H}}\{ T(h):h \geq c+g \}   \\
	\leq & \inf_{h \in \mathcal{H}}\{ T(h):h \geq c \} +  \inf_{h \in \mathcal{H}}\{ T(h):h \geq g \} = U(c) + U(g),
	\end{split}
	\end{equation*}
	since lexicographical ordering holds: $h_1 + h_2 > c+g$ for any $h_1 > c, h_2 > g$. For positive homogeneity, let $ a \in \mathbb{R}^{+}$: then, 
	\begin{equation*}
	\begin{split}
	U(a c) = & \inf_{h \in \mathcal{H}}\{ T(h):h \geq a c \} \\
	= & \inf_{h \in \mathcal{H}}\{ T(a h):h \geq c \} = a \inf_{h \in \mathcal{H}}\{ T(h):h \geq c \} = aU(c).
	\end{split}
	\end{equation*}
	Since $T \leq U$ on $\mathcal{H}$, we can apply Hahn-Banach theorem and extend $T$ from $\mathcal{H}$ to the whole space $C_{L}(\mu)$. Let us denote such extension as $P$ and prove that $P \leq U$ implies positivity of $P$. 
	Assume $P$ is non-positive. Hence there exists a function $c \in C_L(\mu)$ such that $c \geq 0$ and $P(c) < 0$. However, the argument
	\begin{equation*}
	0 < P(-c) \leq U(-c) = \inf_{h \in \mathcal{H}}\{T(h):h \geq -c   \} \leq 0
	\end{equation*}
	leads us to contradiction.
	Let us define a new linear operator $T_f : \{h +t f:t\in\mathbb{R},h \in \mathcal{H}\} \to \mathbb{R}$ such that it coincides with $T$ on $\mathcal{H}:T_{f}\vert_{\mathcal{H}}=T$ and coincides with $U$ at the point $-f: T_{f}(-f) = U(-f)$. By linearity of $T_f$, it follows that $T_{f}(tf) = tU(f)$. Now, notice that for $t \in \mathbb{R}$, we have
	\begin{equation*}
	\begin{split}
	U(t c) = &\inf_{h \in \mathcal{H}}\{T(h):h \geq t c \} = \inf_{h \in \mathcal{H}}\{T(h): -h\leq -t c \} \\
	= &\inf_{h \in \mathcal{H}}\{T(-h):h\leq -t c \} = -\sup_{h \in \mathcal{H}}\{T(h):h\leq -t c \} \\
	\geq& -\inf_{h \in \mathcal{H}}\{T(h):h\geq -t c \} = -U(-t c) = tU(c),
	\end{split}
	\end{equation*}
	where the last inequality follows from the positivity of $T$.
	Thus, $tU(f) \leq U(t f)$ and $T_f \leq U$ everywhere on its domain. By using the Hahn-Banach theorem, we can then extend $T_f$ to the linear functional $P_f : C_{L}(\mu) \to \mathbb{R}$ such that $P_{f}\vert_{\mathcal{H}} = T_f, P_f(-f) = U(-f), P_f \leq U.$ By the construction of linear extensions, we have
	\begin{equation*}
	\sup_{P}P(-f) \leq \inf_{h \in \mathcal{H}}\{T(h):h\geq -f  \}.
	\end{equation*}
	Equivalently,
	\begin{equation*}
	\inf_{P}P(f) \geq \sup_{h \in \mathcal{H}}\{T(h):h \leq f  \}.
	\end{equation*}
	However, using the equality $P_f(-f) = U(-f)$, linearity properties of $T$ and $P$, and the fact that $P_f$ is an extension of $T$ and is dominated by $U$, we finally have
	\begin{equation*}
	\inf_{P}P(f) = \sup_{h \in \mathcal{H}}\{T(h):h \leq f  \}.
	\end{equation*}
	The last equality differs from the desired duality statement by the fact that the infimum is taken over the family of linear operators $P$, which are not measures a priori. Therefore in the remaining part of the proof we  show that actually these functionals are transport plans. Define for any $P$ its restriction $P\vert_{C_{b}} \in (C_{b}(\mathcal{X}))^{*}$ on the dual space of the space of bounded continuous functions on $\mathcal{X}$. 
	As $\mathcal{X}$ is not assumed compact, we need the following result to establish a connection between bounded continuous functions and continuous functions on the Stone-\v{C}ech compactification:

Namely, let $(\mathcal{X},d)$ be a metric space. There exists a compact Hausdorff space $Y$ and a map $T:\mathcal{X} \to \mathcal{Y}$ such that
\begin{itemize}
\item $T$ is a homeomorphism from $\mathcal{X}$ onto $T(\mathcal{X})$,
\item $T(\mathcal{X})$ is dense in $Y$,
\item for every $h \in C_b (\mathcal{X})$ there exists a unique function $g \in C(\mathcal{Y})$ extending $h$ through $T$.
\end{itemize}

The pair $(\mathcal{Y},T)$ is essentially unique and called the \textit{Stone-\v{C}ech compactification }of $\mathcal{X}$. In the following, we  denote it as $\beta \mathcal{X}$.  Since every function $h \in C_b(\mathcal{X})$ can be extended uniquely to $\beta \mathcal{X}$ as a continuous function, there is a natural linear isometry between $C_b(\mathcal{X})$ and $C(\beta \mathcal{X})$. Now, by the Riesz-Markov-Kakutani representation theorem, it follows that $$(C_{b}(\mathcal{X}))^{*} \simeq (C(\beta \mathcal{X}))^{*} \simeq \mathcal{M}(\beta \mathcal{X}),$$ where $\mathcal{M}(\beta \mathcal{X})$ is the space of Radon measures on $\beta \mathcal{X}$. For simplicity, denote as $\pi$ the associated Borel measure on $\beta \mathcal{X}$. By Theorem 2.14 in Rudin \cite{rudin2006real}, since $\Vert P\vert_{C_{b}} \Vert = 1$,  $\pi$ is a probability measure. The next step is to show that the restriction of $\pi$ on $\mathcal{X}$, defined by $\pi\vert_{\mathcal{X}}(A) := \pi(\mathcal{X} \cap A)$, for all $\pi$-measurable sets $A \in \beta \mathcal{X}$, is also a probability a measure. Consider the projection $\Pr_{i} : \mathcal{X} \to \mathcal{X}_{i}$ pushing forward the measure $\pi \vert_{\mathcal{X}}$ to some measure on $\mathcal{X}_{i}$. We have that:
\begin{equation*}
(({\Pr}_{i})_{\#}\pi\vert_{\mathcal{X}})(A_{i}) = \pi\vert_{\mathcal{X}}({\Pr}_{i}^{-1}(A_{i})) = \pi({\Pr}_{i}^{-1}(A_{i})) = \int_{\mathcal{X}_i}\mathds{1}_{A_i}\, d\mu_{i} = \mu_{i}(A_i),
\end{equation*}
for any $\mu_{i}$-measurable set $A_{i}$. Notice that the second-to-last inequality follows from the fact that, for any function $h_{i} \in C_{L}(\mu_{i}),$ which is integrable with respect to $\mu_{i}$, $\int_{\mathcal{X}}h_{i}\, d\pi = \int_{\mathcal{X}_{i}}h_{i}\, d\mu_{i}$. In particular, $\pi \vert_{\mathcal{X}}(\mathcal{X}) = \pi \vert_{\mathcal{X}}({\Pr}_{i}^{-1}(\mathcal{X}_{i})) = \mu_{i}(\mathcal{X}_{i}) = 1$. Thus, we obtained that $P\vert_{C_{b}} \simeq \pi$ is a probability measure on $\mathcal{X}$ with marginals $\mu_{i}$ (that is, a transport plan). The next and final objective is to show that $P$ itself is also a measure. Define the following seminorn on $C_{L}(\mu)$:
\begin{equation*}
\Vert c \Vert_{D} := \inf_{h \in \mathcal{H}}\{T(h) : h \geq \vert c \vert \}.
\end{equation*}
Notice that $\Vert c \Vert_{D} = U(\vert c \vert)$, so it is easy to show that is a well-defined seminorm. Furthermore, the induced topology is stronger than the $\Vert \cdot \Vert_{L}$-topology:
\begin{equation*}
\inf_{h \in \mathcal{H}}\{T(h) : h \geq \vert c \vert \} \geq \sup_{\pi \in \Pi(\mu)}\int_{\mathcal{X}}\vert c \vert \, d\pi = \Vert c \Vert_{L}, 
\end{equation*}
$P$ is continuous with respect to $\Vert c \Vert_{D}$: 
\begin{equation*}
P(\vert c \vert) \leq U(\vert c \vert) = \Vert c \Vert_{D},
\end{equation*}
and $P\vert_{C_{b}}$ is also continuous. To complete the proof of Lemma \ref{Kantorovich duality}, we finally need to show that $C_{b}(\mathcal{X})$ is dense in $C_{L}(\mu)$ with respect to the seminorm  $\Vert c \Vert_{D}$.

To prove this claim, pick $g \in C_{L}(\mu)$, and let $\vert g \vert \leq h \in \mathcal{H}$. Let $k \in \mathbb{N}$, $g^{k} := \min\{k,g \}$ and $g^{k}_{k} := \max\{ g^{k},-k \}, g^{k}_{k} \in C_{b}(\mathcal{X})$. We need to show that $\Vert g - g^{k}_{k} \Vert_{D} \to 0$ as $k \to \infty$. First, notice that
	\begin{equation*}
	\Vert g - g^{k}_{k} \Vert_{D} \leq  \Vert g - g^{k} \Vert_{D} +  \Vert g^{k} - g^{k}_{k} \Vert_{D}.
	\end{equation*}
	Then, by Lebesgue dominated convergence theorem,
	\begin{equation*}
	\Vert g - g^{k} \Vert_{D} = U(\max\{g-k,0\}) \leq U\left( \sum_{i=1}^{n}\max\left\{h_{i}-\dfrac{k}{n},0\right\}  \right) \to 0, \quad \mbox{as } k \to \infty,
	\end{equation*}
	\begin{equation*}
	\Vert g^{k} - g_{k}^{k} \Vert_{D} = U(\max\{-g^{k}-k,0\}) \leq U\left( \sum_{i=1}^{n}\max\left\{h_{i}-\dfrac{k}{n},0\right\}  \right) \to 0, \quad \mbox{as } k \to \infty,
	\end{equation*} 
	with $U$ defined as in (\ref{U}). This shows that $C_{b}(\mathcal{X})$ is dense in $C_{L}(\mu)$ with respect to the seminorm  $\Vert c \Vert_{D}$.
		
Also, $P\vert_{C_{b}}$ can be seen as a linear operator on $C_{L}(\mu)$, so $P\vert_{C_{b}}$ and $P$ are both continuous linear functionals on $ (C_{L}(\mu),\Vert \cdot \Vert_{D})$ and coincide on $C_{b}(\mathcal{X})$. Using the fact that $C_{b}(\mathcal{X})$ is dense in $ (C_{L}(\mu),\Vert \cdot \Vert_{D})$, we obtain that $P\vert_{C_{b}}$ extended on $C_{L}(\mu)$ and $P$ also coincide on the whole $C_{L}(\mu)$:
\begin{equation*}
P \simeq \pi \in \Pi(\mu).
\end{equation*}
Finally, by noticing that $P$ is a transport plan with marginals $\mu_{i}$, the proof of Lemma \ref{Kantorovich duality} is complete.
\end{proof}
\begin{lem}[Sion's minimax theorem] \label{minimax}
Let $K$ be a compact convex subset of a Hausdorff topological vector space, $Y$ be a convex subset of an arbitrary vector space, and $h$ be a real-valued function on $K \times Y$, which is lower semi-continuous and convex on $K$ for each fixed $y \in Y$, and concave on $Y$. Then 
\begin{equation*}
\min_{x \in K} \sup_{y \in Y} h(x,y) = \sup_{y \in Y}\min_{x \in K}h(x,y).
\end{equation*}
\end{lem}
\begin{proof}
See Adams and Hedberg \cite{adams2012function}.
\end{proof}
\subsection{Proof of Theorem \ref{constrained Kantorovich duality}}  \label{proof-Kantorovich-constraints} 
\begin{proof}	The first inequality is easy to prove:
	\begin{equation*}
	\begin{split}
	\inf_{\pi \in \Pi_{W}(\mu)}\int_{\mathcal{X}} f \, d\pi & \geq \inf_{\pi \in \Pi_{W}(\mu)} \sup_{\substack{h + w \leq f \\ h \in \mathcal{H}, \, w \in W}} \int_{\mathcal{X}} (h+w)d\pi  \\
	& = \inf_{\pi \in \Pi_{W}(\mu)} \sup_{\substack{h + w \leq f \\ h \in \mathcal{H}, \, w \in W}} \sum_{i=1}^{n} \int_{\mathcal{X}_i} h_{i}\, d\mu_{i} = \sup_{\substack{h + w \leq f \\ h \in \mathcal{H}, \, w \in W}} \sum_{i=1}^{n} \int_{\mathcal{X}_i} h_{i}\, d\mu_{i}.
	\end{split}
	\end{equation*}
	The second inequality is much more difficult. First, notice that
	\begin{equation} \label{step_before_KD}
	\sup_{\substack{h + w \leq f \\ h \in \mathcal{H}, \, w \in W}} \sum_{i=1}^{n}\int_{\mathcal{X}_i} h_{i}\, d\mu_{i} = \sup_{w \in W} \sup_{\substack{h \leq (f - w) \\ h \in \mathcal{H}}}\sum_{i=1}^{n}\int_{\mathcal{X}_i} h_{i} \, d\mu_{i}.
	\end{equation}
Applying Lemma \ref{Kantorovich duality} for the function $(f-w) \in C_{L}(\mu)$ in (\ref{step_before_KD}), we have: 
\begin{equation*}
\sup_{w \in W} \sup_{\substack{h \leq (f - w) \\ h \in \mathcal{H}}}\sum_{i=1}^{n}\int_{\mathcal{X}_i} h_{i} \, d\mu_{i} = \sup_{w \in W}\inf_{\pi \in \Pi(\mu)}\int_{\mathcal{X}}(f-w)d\pi.
\end{equation*}
By linearity and continuity of $\int_{\mathcal{X}}(f-w)d\pi$, in order to safely interchange supremum and infimum, we can then apply Lemma \ref{minimax} for $ K = \Pi(\mu)$, $Y = W$ and $h(x,y) = \int_{X} (f-w) \, d\pi$, obtaining:
\begin{equation*}
\sup_{w \in W}\inf_{\pi \in \Pi(\mu)}\int_{\mathcal{X}}(f-w)d\pi = \inf_{\pi \in \Pi(\mu)} \sup_{w \in W} \int_{\mathcal{X}}(f-w)d\pi.
\end{equation*}
Finally, notice that if $\pi \notin \Pi_{W}(\mu)$, then there exists $w_1 \in W$ such that $\int_{\mathcal{X}}w_1 d\pi < 0$. Then, by choosing $w = \alpha w_1 $, $\alpha \to +\infty$, this implies $\sup_{w \in W} \int_{\mathcal{X}}(f-w) = +\infty$. Thus, we can conclude:
\begin{equation*}
\inf_{\pi \in \Pi(\mu)} \sup_{w \in W} \int_{\mathcal{X}}(f-w)d\pi = \inf_{\pi \in \Pi_W(\mu)} \int_{\mathcal{X}}f\,d\pi,
\end{equation*}
	which is what we need to complete the proof of Theorem \ref{constrained Kantorovich duality}.
\end{proof}

\section{Further numerical results} \label{further_numerics}
The objective of this section is twofold: first, we aim to provide a sensitivity analysis of Algorithm \ref{OT via NN: AGD} with respect to the choice of the penalization coefficient $\gamma$, the training batch size and the width (hidden dimension) of the hidden layers; second, we will show instances of the (primal) optimizers  obtained via Eq. (\ref{joint_distribution}). In order to focus on these aspects and ease the visualization, we will study a simple 2D example with one constraint. 

More precisely, let us consider a basket option with strike $K=0$ on two assets, $S_{A1} \sim \mathcal{N}(0,2), S_{A2} \sim \mathcal{N}(1,2)$, whose price $p_{A1}$ is assumed to be known. Of course, from a theoretical point of view, this additional information should constrain the upper and lower bounds on the basket option to match exactly the given price.  

In Figures \ref{fig:sens_gamma}-\ref{fig:sens_neurons}, we report the relative difference between numerical and exact values (denoted as $\Delta\%$Bounds) for different penalization coefficients $\gamma$, batch sizes and numbers of neurons per hidden layer, respectively. Overall, we found that $\gamma=100$ worked satisfactorily well across all experiments while being less prone to instability during training than, say, $\gamma=400$ or $1000$. Similarly, a batch size equal to $2^8$ and a number of hidden neurons equal to $128$ seemed to work quite well for our purposes.\footnote{Of course, we would like to remark that these conclusions might differ for much higher dimensional settings (e.g., $d > 10$), where more fine-tuning of the hyperparameters could be required.}

\begin{figure}[!h]
	\centering
 \includegraphics[width=0.8\textwidth,scale = 0.6]{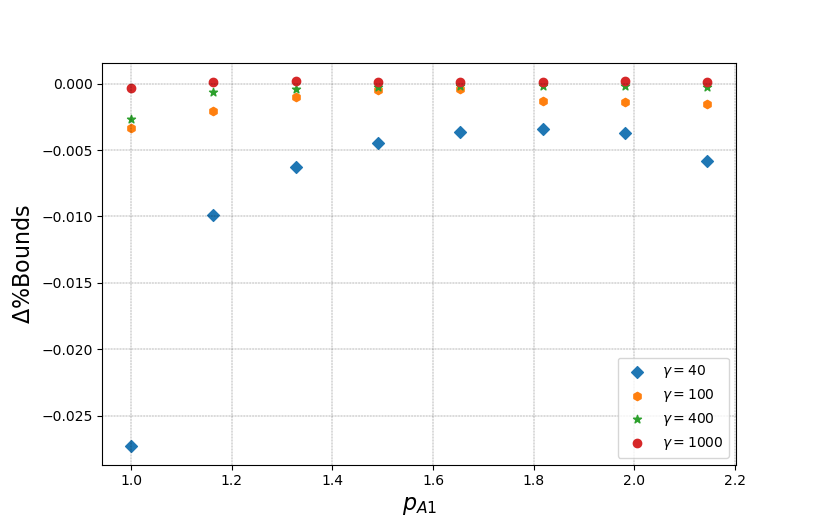} 
	\caption{2-basket option with one constraint for different penalization coefficients $\gamma$. The batch size is set as $2^8$ and the number of hidden neurons as 128.} \label{fig:sens_gamma}
\end{figure}

\begin{figure}[!h]
	\centering
	\includegraphics[width=0.8\textwidth,scale = 0.6]{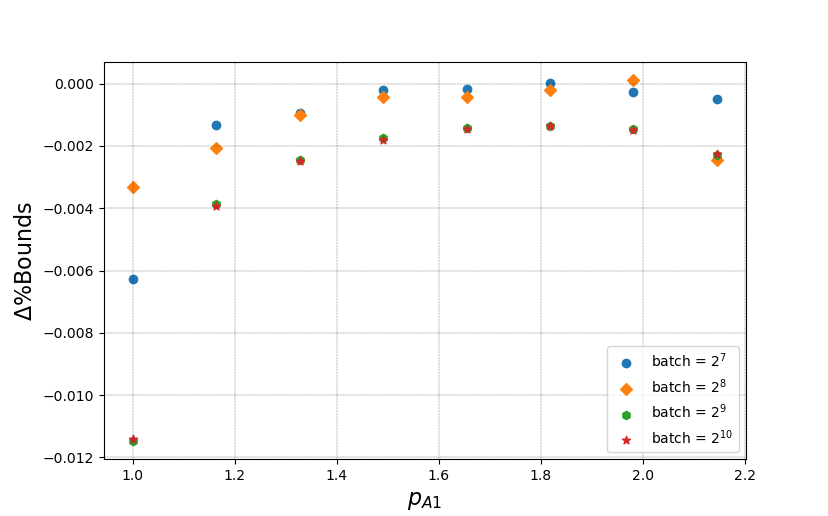} 
	\caption{2-basket option with one constraint for different training batch sizes. $\gamma$ is set as $100$ and the number of hidden neurons as 128.} \label{fig:sens_batch}
\end{figure}

\begin{figure}[!h]
	\centering
	\includegraphics[width=0.8\textwidth,scale = 0.6]{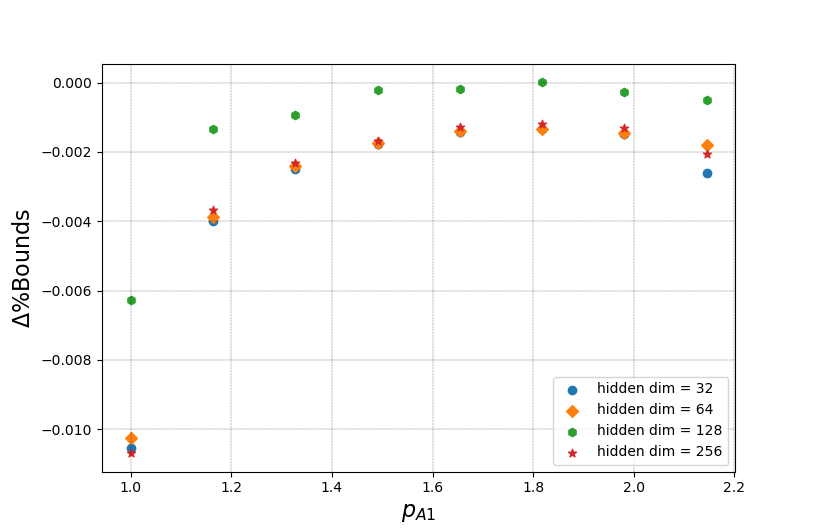} 
	\caption{2-basket option with one constraint for different numbers of hidden neurons. $\gamma$ is set as $100$ and the batch size as $2^{8}$.} \label{fig:sens_neurons}
\end{figure}

To conclude, in Figure \ref{fig:copulas}, we show some instances of the optimizers (couplings) that are obtained by utilizing the formula in Eq. (\ref{joint_distribution}) and the dual solution. The numerical optimal measures are displayed (at the copula level) for four values of $p_{A1}$. 

Figures \ref{fig:cop1} and \ref{fig:cop4} correspond to the extreme cases. In particular, Figure \ref{fig:cop1} corresponds to the copula that minimizes the price of the basket payoff $(p_{A1} = 1)$. Notice that, for this problem, the solution is not unique. In fact, one could have rightfully expected an antimonotonic copula. However, the algorithm finds an alternative minimizing dependence with negative correlation $\rho = -0.92$. 

On the other hand, Figure \ref{fig:cop4} corresponds to the copula that maximizes the price of the payoff $(p_{A1} = 2.145)$. Similarly, instead of a comonotonic structure, here the algorithm finds a different maximizing dependence with positive correlation $\rho = 0.63$.  

Finally, Figures \ref{fig:cop2} and \ref{fig:cop3} represent two intermediate situations. For $p_{A1} = 1.4$, the optimal measure between the marginals shows a mild negative correlation $(\rho = -0.29)$; for $p_{A1} = 1.8$, we are close to the independent coupling $(\rho = 0.02)$.

\begin{figure}
\centering
\begin{subfigure}[b]{0.475\textwidth}
	\centering
	\includegraphics[width=\textwidth]{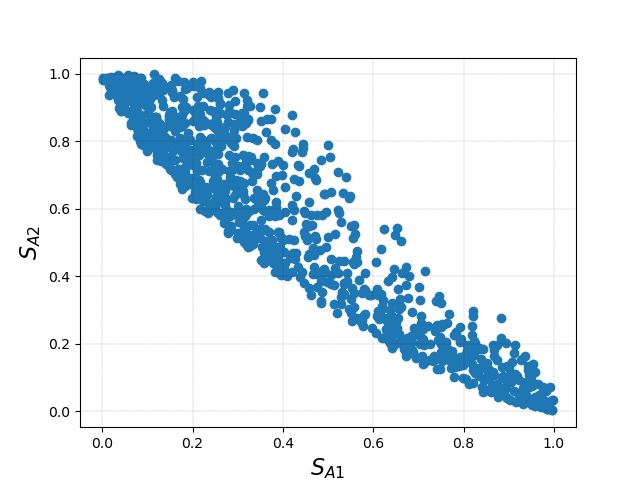}
	\caption[]{$p_{A1} = 1$}
	\label{fig:cop1}
\end{subfigure}
\hfill
\begin{subfigure}[b]{0.475\textwidth}  
	\centering 
	\includegraphics[width=\textwidth]{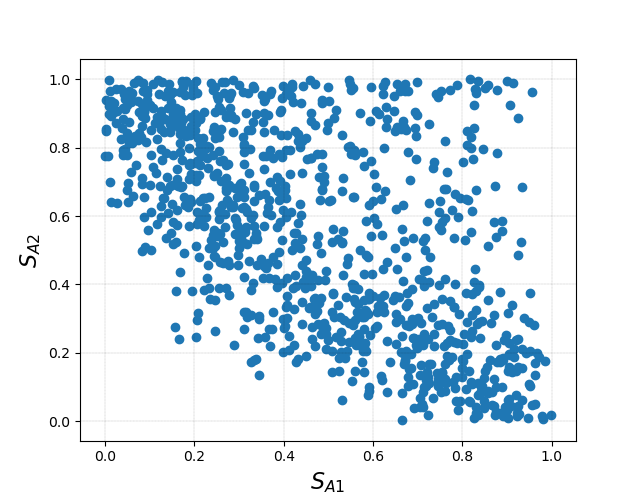}
		\caption[]{$p_{A1} = 1.4$}
	\label{fig:cop2}
\end{subfigure}
\vskip\baselineskip
\begin{subfigure}[b]{0.475\textwidth}   
	\centering 
	\includegraphics[width=\textwidth]{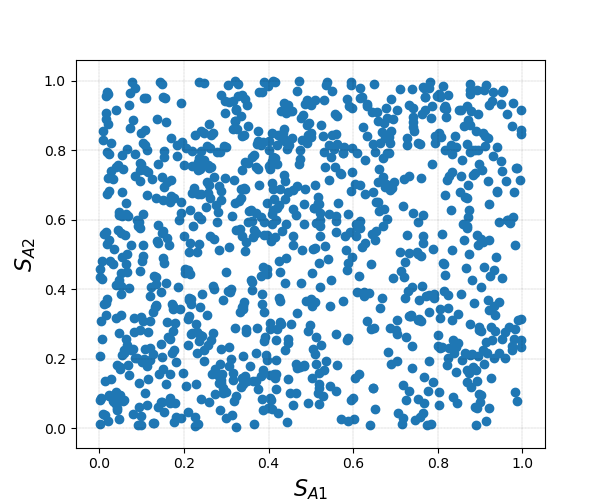}
		\caption[]{$p_{A1} = 1.8$}  
	\label{fig:cop3}
\end{subfigure}
\hfill
\begin{subfigure}[b]{0.475\textwidth}   
	\centering 
	\includegraphics[width=\textwidth]{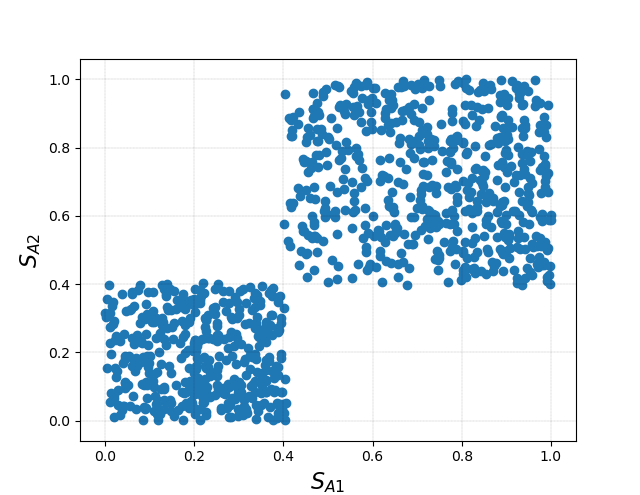}
		\caption[]{$p_{A1} = 2.145$}
	\label{fig:cop4}
\end{subfigure}
	\caption[]{Copula between $S_{A1}$ and $S_{A2}$ for different values of $p_{A1}$.}
\label{fig:copulas}
\end{figure}

\end{appendices}	

\end{document}